\documentclass[US letter, 9 pt, conference]{ieeeconf}  
\usepackage{setspace}
\usepackage{graphics} 
\usepackage{graphicx}

\usepackage{epsfig} 
\usepackage{float}
\usepackage{times} 
\usepackage{stdmath}
\usepackage{amsmath}  
\usepackage{amsfonts}
\newtheorem{theorem}{Theorem}
\newtheorem{definition}{Definition}
\newtheorem{lemma}{Lemma}

\usepackage{graphicx}
\usepackage{caption}
\usepackage{subcaption}
\makeatother

\newcommand{\pfx}{\frac{\partial}{\partial x}}
\newcommand{\pfxi}{\frac{\partial}{\partial \xi}}

\newcommand{\igzo}{\int_0^1}
\newcommand{\igzx}{\int_0^x}
\newcommand{\igxo}{\int_x^1}

\newcommand{\wh}{\hat{w}}

\newcommand{\lt}{L_2(0,1)}

\newcommand{\hlf}{\frac{1}{2}}
\newcommand{\mcl}[1]{\mathcal{#1}}
\newcommand{\pop}{\mathcal{P}}
\newcommand{\pinv}{\mathcal{P}^{-1}}
\newcommand{\sop}{\mathcal{S}}
\newcommand{\sinv}{\mathcal{S}^{-1}}
\newcommand{\zh}{\hat{z}}
\allowdisplaybreaks

\IEEEoverridecommandlockouts                              

\title{\LARGE \bf
Output Feedback Control of Inhomogeneous Parabolic PDEs with Point Actuation and Point Measurement using SOS and Semi-Separable Kernels
}

\author{Aditya Gahlawat$^{1}$ and Matthew M. Peet$^{2}$
\thanks{$^{1}$Aditya Gahlawat is with the Department
of Mechanical, Materials and Aerospace Engineering at the Illinois Institute of Technology, Chicago,
IL, 60616 USA and with the Grenoble Image Parole Signal Automatique Lab., Universit\'{e} Joseph Fourier/Centre National de la Recherche Scientifique, St. Martin d'Heres, France
        {\tt\small agahlawa@hawk.iit.edu}}%
\thanks{$^{2}$Matthew. M. Peet is with the School of Engineering of Matter, Transport and Energy at Arizona State University, Tempe, AZ, 85287-6106 USA
        {\tt\small mpeet@asu.edu}}%
}

\begin{document}

\maketitle
\thispagestyle{empty}
\pagestyle{empty}

\begin{abstract}
In this paper we use SOS and SDP to design output feedback controllers for a class of one-dimensional parabolic partial differential equations with point measurements and point actuation. Our approach is based on the use of SOS to search for positive quadratic Lyapunov functions, controllers and observers. These Lyapunov functions, controllers and observers are parameterized by linear operators which are defined by SOS polynomials. The main result of the paper is the development of an improved class of observer-based controllers and evidence which indicates that when the system is controllable and observable, these methods will find a observer-based controller for sufficiently high polynomial degree (similar to well-known results from backstepping).
\end{abstract}

\section{INTRODUCTION}

Parabolic Partial Differential Equations (PDEs) are a class of system used to model processes such as diffusion, transport and reaction.  Some examples of systems which have been modelled using parabolic PDEs include plasma in a tokamak~\cite{witrant2007control}, heat propagation, and spatial dynamics of population in an ecosystem~\cite{murray2002mathematical}. In this paper we consider a class of inhomogeneous linear scalar valued Parabolic PDEs. We assume that only boundary control and sensing is available for the PDEs. The control is exercised through a Dirichlet boundary condition and a Neumann boundary measurement of the state is available. The goal of this article is to use this boundary measurement to construct a boundary controller which ensures that the state of the system remains bounded in the presence of an exogenous input (has finite $L_2$-gain). We refer to this as output feedback based boundary control.

In order to design output feedback based boundary controllers, we design a Luenberger observer where the error dynamics have finite $L_2$ gain from disturbance to error. We then design a controller for the system which utilizes the state of the observer and show that for the resulting closed-loop system, there is a bound on the $L_2$ gain from disturbance to output. Our approach is based on parameterizing the set of quadratic Lyapunov functions by the set positive operators, which in turn is parameterized by polynomials and ultimately by Sum-of-Squares (SOS) polynomials and positive matrices - leading to a Linear Matrix Inequality (LMI). The approach we take in this paper is akin to LMI methods for control of linear Ordinary Differential Equations (ODEs) using Lyapunov inequalities and a variable substitution trick. However, because our inequalities are expressed as operators in Hilbert space, we refer to our approach as a Linear Operator Inequality (LOI).

This article extends our work in~\cite{gahlawat2011designing} wherein we designed output feedback boundary controllers for a one-dimensional homogenous heat equation by considering a simpler class of positive operators (also parameterized by SOS polynomials). This paper improves on the work in~\cite{gahlawat2011designing} by a) Considering the larger class of inhomogeneous, possibly unstable parabolic PDEs b) By considering a larger class of Lyapunov functions defined by positive multiplier and integral operators with semi-separable kernels and c) providing evidence (but not a proof) that this new class of operators can be used to design output-feedback based controllers whenever the system is observable and controllable. Specifically, the class of Lyapunov functions we use has the form
\begin{align*}
V(w)=&\igzo M(x)w(x)^2dx+\igzo\igzo w(x)K(x,\xi)w(\xi)d\xi dx,
\end{align*} where
\[K(x,\xi) = \begin{cases} K_1(x,\xi) & \xi \leq x \\
K_2(x,\xi) &  \xi>x \end{cases} \]
and were $M$, $K_1$ and $K_2$ are polynomials and $w$ represents the spatially distributed state of the PDE. A kernel $K$ of this form is referred to as \textit{semi-separable}.

One popular and relatively straightforward method for output feedback boundary control of PDEs is backstepping~\cite{krstic2008boundary}. This method relies on constructing an invertible operator which, in closed loop, maps the state of the system to the state of a chosen stable system for which a quadratic Lyapunov function exists. Our approach varies in the fact that we search for both the controller and the quadratic Lyapunov function. Although our approach is different, similar to backstepping, the numerical results indicate that we can construct output feedback controllers for any controllable and observable system. Some other examples of work which use Lyapunov functions for analysis and control of PDEs are~\cite{coron2008dissipative}, \cite{coron2007strict}. An example of application of LMIs for the control of PDEs is~\cite{fridman2009lmi} where the authors synthesize stabilizing boundary controllers for uncertain semi-linear PDEs using quadratic Lyapunov functions parameterized by positive scalars. Early results on the use of SOS for analysis and control of infinite-dimensional systems can be found in~\cite{peet2006positive}, \cite{papachristodoulou2006analysis}. Additional recent work on the application of polynomials to infinite-dimensional systems can be found in the research done by our colleagues in~\cite{Valmo_1} and~\cite{Valmo_2}.

The paper is organized as follows: Section~\ref{sec:pro_state} outlines the problem statement and presents background on SOS polynomials. In Section~\ref{sec:posop} we define the class of positive operators which we utilize. Section~\ref{sec:prelim} provides a controller synthesis condition and related inequalities which are later used to prove the main result. In Sections~\ref{sec:control} we provide the main results wherein we construct output feedback controllers. Section~\ref{sec:num_results} provides the numerical results for an example PDE.
\section{NOTATION}\label{sec:notation}

$\R^{m \times n}$ denotes the set of real $m$-by-$n$ matrices. $\S^n\subset \R^{n \times n}$ is the subspace of symmetric matrices. $I_n$ is the identity matrix of dimension $n \times n$ and we denote $I=I_n$ when $n$ is clear from context. For any $\Omega \subset \R$, $C^m(\Omega)$ is the space of $m$-times continuously differentiable functions defined on $\Omega$. Similarly, for any $\Omega_1,\Omega_2 \subset \R$, $C^{m,n}(\Omega_1 \times \Omega_1)$ is the space of functions which are $m$-times and $n$-times continuously differentiable on $\Omega_1$ and $\Omega_2$ respectively.  The shorthand $u_x$ denotes the partial derivative of $u$ with respect to independent variable $x$. We use $\lt$ to denote the Hilbert space of square integrable functions from $[0,1]$ to $\R$.
Unless otherwise indicated, $\langle \cdot,\cdot \rangle$ denotes the inner product on $\lt$ and $\|\cdot\|=\|\cdot\|_{\lt}$ denotes the norm induced by the inner product. Similarly, $L_2(0,\infty;\lt)$ denotes the Hilbert space of square integrable functions from $[0,\infty)$ to $\lt$ equipped with the norm
\[\norm{f}_{L_2(0,\infty;\lt)}=\left(\int_0^\infty \norm{f(\cdot,t)}^2 dt \right)^{\frac{1}{2}}.\]
$H^n(0,1):=\{y \in L_2 :\frac{d^{i}y}{dt^{i}}\in \lt,\,i=1,\cdots n\}$ is the Sobolev subspace equipped with inner product $\ip{x}{y}_{H^n} = \sum_{m=0}^n\ip{\frac{d^m x}{dt^m}}{\frac{d^m y}{dt^m}}$. For Hilbert spaces $X$ and $Y$, the set $\mcl{L}(X,Y)$ is the Banach space of bounded linear operators from $X$ to $Y$ endowed with the induced norm $\|\cdot\|_\mathcal{L}$. $\mcl{I}$ denotes the identity operator.
We define $Z_d(x)$ to be the column vector of all monomials in variables $x$ of degree $d$ or less. For brevity, we sometimes use $Z_d(x,\xi) = Z_d([x;\xi])$.

\section{PROBLEM STATEMENT}\label{sec:pro_state}
In this paper, we consider the following scalar parabolic PDE
\begin{equation}\label{eqn:prob:PDE_form}
w_t(x,t)=a(x)w_{xx}(x,t)+b(x)w_x(x,t)+c(x)w(x,t)+f(x,t),
\end{equation}
where $x\in [0,1]$, $t \geq 0$, with mixed boundary conditions of the form
\begin{equation}\label{eqn:prob:PDE_form_BC}
w_x(0,t)=0, \quad w(1,t)=u(t),
\end{equation} Here $a$, $b$ and $c$ are polynomials with $a(x) \geq \alpha >0$, for $x \in [0,1]$. Additionally, $f \in L_2\left(0,\infty;\lt\right)$ is the \textit{exogenous input} and $u(t)$ is the \textit{control input}. The output of the system is $y(t)=w_x(1,t)$. Note that we have also considered several other types of observer-controller boundary conditions, which will be listed in the section on numerical results. The first goal of the paper is to find a control operator $\mcl{F}\in \mcl{L}\left(H^2(0,1), \R \right)$ such that if $u(t)=\mcl{F}w(\cdot,t)$, then the closed-loop PDE system is stable.

Next, using the Luenberger framework, we assume our observer has the form
\begin{align}
\wh_t(x,t)=&a(x)\wh_{xx}(x,t)+b(x)\wh_x(x,t)+c(x)\wh(x,t) \nonumber \\
&\label{eqn:prob:observer_form}+O_1(x)\left(\hat{y}(t)-y(t) \right),
\end{align} with boundary conditions
\begin{equation}\label{eqn:prob:observer_form_BC}
\hat{w}_x(0,t)=0, \quad \wh(1,t)=u(t)+O_2\left(\hat{y}(t)-y(t) \right),
\end{equation}
where the function $O_1(x)$ and scalar $O_2$ must be chosen such that the dynamics of the error $e(x,t)=w(x,t)-\wh(x,t)$ are stable. The second goal of the paper, then, is to find such $O_1(x)$ and $O_2$ and show that if $u(t)=\mcl{F}\wh(\cdot,t)$, then the coupled system of parabolic PDEs is stable and satisfies
\[\norm{w}_{L_2\left(0,\infty;\lt\right)}  \leq \gamma \norm{f}_{L_2\left(0,\infty;\lt\right)}.\]
for some $\gamma>0$. Note that for the system and the observer, we assume the existence of classical solutions belonging to $C^{1,2}\left((0,\infty) \times [0,1] \right)$. This assumption can be validated using the analysis presented in~\cite{balogh2004stability} and~\cite{fridman2009lmi}.

\subsection{SOS and Operators}\label{subsec:sos}
SOS is an approach to  the optimization of positive polynomial variables. Given a polynomial $f(y)$, $y\in \R^n$, the feasibility problem of determining if the polynomial is globally positive ($f(y)\ge0$ for all $y\in \R^n$) is NP-hard~\cite{blum1998complexity}. To overcome this difficulty, there are a number of sufficient conditions for polynomial positivity. A particularly important such condition is that the polynomial, $p$, be a Sum-of-Squares (SOS), so that $p\,(x)=\sum_{i=1}^k g_i(x)^2$, for polynomials $g_i$  and which is denoted $p \in\Sigma_s$. The importance of the SOS condition lies in the fact that it can be readily enforced using LMIs. This is due to the easily proven fact that for a polynomial $p$ of degree $2d$, $p \in \Sigma_s$ if and only if $p=Z(x)^T Q Z(x)$ for some $Q\ge 0$, where $Z(x)$ is the vector of monomials of degree $d$ or less~\cite{parrilo2000structured}.  A recent survey for alternatives to SOS based methods may be found in~\cite{kamyar2014polynomial}.

We can use SOS to construct positive operators on $\lt$. For example, define the operator $\pop z(x)=M(x)z(x)$, $z\in \lt$, where $M$ is a polynomial. If, for $\epsilon>0$, $M(x)-\epsilon \in \Sigma_s$, then the operator $\pop$ is positive on $\lt$. Therefore, we may conclude that $\pop$ is positive on $\lt$ if there exists a $Q>0$ such that
 $M(x)-\epsilon=Z(x)^T Q Z(x)$. By equating the coefficients on the left and right-hand sides, we obtain an LMI test for positivity of the operator. Of course, the operators considered in this paper are significantly more complicated than $\pop z$.
\section{POSITIVE OPERATORS ON $\lt$}\label{sec:posop}
In this paper, our results are expressed as optimization over a set of positive operators. To solve these optimization problems, we use positive matrices to parameterize a subset of positive operators on $\lt$ as described in~\cite{peetlmi}. Specifically, we consider operators of the form
\begin{equation}\label{eqn:Poperator}
(\mathcal{P}z)(x)=M(x)z(x) +  \int_0^1 K(x,\xi)z(\xi)d\xi, \quad z\in \lt,
\end{equation} with semi-separable kernel
\[
K(x,\xi) = \begin{cases} K_1(x,\xi) & \xi \leq x \\
K_2(x,\xi) &  \xi>x \end{cases} ,\]
where $M:[0,1] \rightarrow \R$ and $K_1,K_2: [0,1] \times [0,1] \rightarrow \R$ are polynomials.  In~\cite{peet2008using}, we gave necessary and sufficient conditions for positivity of multiplier and integral operators of similar form using pointwise constraints on the functions $M$, $K_1$ and $K_2$. Recently, in~\cite{peetlmi}, these conditions were sharpened - See Theorem~$1$. The following theorem is an extension of this result.
\begin{theorem}\label{thm:jointpos}
Let
\begin{align*}
M(x) =& Z_{1}(x)^T U_{11}Z_{1}(x),\\
K_1(x,\xi) = &Z_{1}(x)^T U_{12}Z_{2}(x,\xi) + Z_{2}(\xi,x)^T U_{31}Z_1(\xi)\\
&+\int_0^\xi Z_{2}(\eta,x)^T U_{33}Z_{2}(\eta,\xi)d\eta  \\
&+\int_\xi^x Z_{2}(\eta,x)^T U_{32}Z_{2}(\eta,\xi)d \eta \\
& +\int_x^1 Z_{2}(\eta,x)^T U_{22}Z_{2}(\eta,\xi)d\eta,
\end{align*}
where $K_2(x,\xi) = K_1(\xi,x)$, $Z_1(x) = Z_{d_1}(x)$ and $Z_2(x,y) = Z_{d_2}(x,y)$ and
\begin{equation}\label{eqn:posop1}
U=\left[\begin{array}{ccc} U_{11} & U_{12} & U_{13} \\
U_{21} & U_{22} & U_{23} \\
U_{31} & U_{32} & U_{33}
\end{array} \right] \ge \bmat{\epsilon_1 I&0&0\\0&0&0\\0&0&0},
\end{equation}
Then the operator $\mathcal{P}$ defined in Eqn.~\eqref{eqn:Poperator} is  self-adjoint and satisfies
\[\epsilon_1 \|z\|^2 \leq \langle \mathcal{P}z,z \rangle \leq \epsilon_2 \|z\|^2, \text{ for all } z \in \lt.\]
where $\epsilon_2 = (\theta_1+\theta_2) \lambda_{\max}(U)$, $\lambda_{\max}(U)$ is the maximum eigenvalue of $U$, and
\begin{alignat*}{2}
\theta_1& &&= \sup_{x\in [0,1]} Z_1(x)^T Z_1(x),\\
\theta_2& &&= \sup_{(x,\xi)\in [0,1] \times [0,1]}\left|\int_0^\xi g(x,\xi,\eta) d\eta+\int_x^1 g(x,\xi,\eta) d\eta \right|,\\
g(x,\xi,\eta)& &&= Z_2(\eta,x)^T Z_2(\eta,\xi).
\end{alignat*}
\end{theorem}
\begin{proof}
The proof is based on the result in~\cite{peetlmi} and is omitted for brevity.
\end{proof}
For convenience, we define the set of multipliers and kernels which satisfy Theorem~\ref{thm:jointpos}.
\begin{align*}
 \Xi_{\{d_1,d_2,\epsilon_1,\epsilon_2\}} =\{ M,K_1,K_2 \, : \, M,K_1,K_2 \text{ satisfy}\\
  \text{Theorem~\ref{thm:jointpos} for $d_1,d_2,\epsilon_1,\epsilon_2$.}\}
\end{align*}
Of course, since such operators are positive definite and bounded on $\lt$, the inverse of these operators exist and are bounded~\cite{kreyszig1989introductory}. However, as will become apparent in subsequent sections, we need a method of constructing the inverse of operators defined by elements of $\Xi_{\{d_1,d_2,\epsilon_1,\epsilon_2\}}$. Fortunately, such methods do exist in literature and we use one such method. Using the terminology presented in~\cite{gohberg1984time} it can be shown that the operators defined by $\Xi_{\{d_1,d_2,\epsilon_1,\epsilon_2\}}$ are the input-output maps of well-posed Linear Time Varying (LTV) systems. For this class of operators, the inverse can be constructed as explained in~\cite{gohberg1984time}.
\section{PRELIMNARY INEQUALITIES}\label{sec:prelim}
In this section we provide a couple of inequalities which we will use for the controller and observer synthesis. We begin by defining the operator $\mcl{A}:H^2(0,1) \rightarrow \lt$ (infinitesimal genearator) which defines the class of PDEs under consideration.
\begin{equation}\label{eqn:A_prelim}
\mcl{A}=a(x)\frac{d^2}{dx^2} +b(x)\frac{d}{dx} +c(x),
\end{equation} where recall $a$, $b$ and $c$ are polynomial functions and $a(x) \geq \alpha >0$, for $x \in [0,1]$. Before presenting the inequalities, we define a pair of mappings which relate the functions $M,K_1,K_2$ to the derivative of the Lyapunov function $V=\ip{w}{\mathcal{P}w}$. The first mapping considers $\ip{\mcl{AP}z}{z}+\ip{z}{\mcl{AP}z}$.
\begin{definition}\label{def:dual}
For scalar $\epsilon_1>0$ and polynomials $a$, $b$ and $c$ which define the PDE under consideration, we say $\{T_0,T_1,T_2,T_3,T_4,T_5,T_6\}=\mcl{M}_{\epsilon_1}\left(M,K_1,K_2\right)$ if
\begin{align*}
T_0(x)=&\left(a_{xx}(x)-b_x(x) \right)M(x)+b(x)M_x(x) \\
&+a(x)M_{xx}(x)+2c(x)M(x)-\frac{\pi^2}{2}\alpha \epsilon_1\\
&+a(x) \left[2\pfx \left[K_1(x,\xi)-K_2(x,\xi) \right] \right]_{\xi=x}, \\
T_1(x,\xi)=&a(x)K_{1,xx}(x,\xi)+b(x)K_{1,x}(x,\xi) \\
&+a(\xi)K_{1,\xi\xi}(x,\xi)+b(\xi)K_{1,\xi}(x,\xi)\\
&+\left(c(x)+c(\xi) \right)K_1(x,\xi),\\
T_2(x,\xi)=&T_1(\xi,x),\\
T_3=&\left(a_x(0)-b(0) \right)M(0)+a(0)M(0)-\frac{\pi^2}{2}\alpha \epsilon_1,\\
T_4(x)=&2a(0)K_{2,x}(0,x)+\pi^2 \alpha \epsilon_1,\\
T_5=&\left(b(1)-a_x(1) \right)M(1)+a(1)M_x(1) ,\\
T_6=&2a(1)M(1),\\
K_{1,x}(1,x)=& \left[K_{1,x}(x,\xi)|_{x=1} \right]_{\xi=x}.
\end{align*}
\end{definition}
The second mapping relates the functions $N,P_1,P_2$ to the derivative of the \textit{dual} functional $\ip{\mcl{A}w}{\sop w}+\ip{\mcl{SA}w}{w}$.
\begin{definition}\label{def:primal}
Given scalar $\epsilon_1>0$ and polynomials $a$, $b$ and $c$ which define the PDE under consideration, we say $\{Q_0,Q_1,Q_2,Q_3,Q_4,Q_5,Q_6,Q_7,Q_8\}=\mcl{N}_{\epsilon_1}\left(N,P_1,P_2\right)$ if the following hold
\begin{align*}
Q_0(x)=& \pfx \left[\pfx a(x)N(x)-b(x)N(x) \right]\\
&+ 2\left[\pfx \left[a(x)\left(P_1(x,\xi)-P_2(x,\xi) \right) \right] \right]_{\xi=x}\\
&+2N(x)c(x)-\frac{\pi^2}{2}\alpha \epsilon_1,\\
Q_1(x,\xi)=&\pfx \left[\pfx a(x)P_1(x,\xi)-b(x)P_1(x,\xi) \right] \\
&+\pfxi \left[\pfxi a(\xi)P_1(x,\xi)-b(\xi)P_1(x,\xi) \right]\\
&+\left(c(x)+c(\xi) \right)P_1(x,\xi),\\
Q_2(x,\xi)=&Q_1(\xi,x),\\
Q_3=& \left(a_x(0)-b(0) \right)N(0)+a(0)N_x(0)-\frac{\pi^2}{2}\alpha \epsilon_1,\\
Q_4(x)=& 2\left(a_x(0)-b(0)\right)P_2(0,x)+2a(0)P_{2,x}(0,x)\\
&+\pi^2 \alpha \epsilon_1,\\
Q_5=&(b(1)-a_x(1))N(1)-a(1)N_x(1),\\
Q_6(x)=&2(b(1)-a_x(1))P_1(1,x)-2a(1)P_{1,x}(1,x),\\
Q_7=&2a(1)N(1),\\
Q_8(x)=&2a(1)P_1(1,x),\\
P_{1,x}(1,x)=& \left[P_{1,x}(x,\xi)|_{x=1} \right]_{\xi=x}.
\end{align*}
\end{definition}

The proofs of the following lemmas are provided in the appendix.

 The first allows us to represent $\dot V=\ip{\mcl{AP}z}{z}+\ip{z}{\mcl{AP}z}$
\begin{lemma}\label{lem:dual_LOI}
For any $\{M,K_1,K_2\} \in \Xi_{\{d_1,d_2,\epsilon_2,\epsilon_2\}}$, $0<\epsilon_1<\epsilon_2<\infty$, let $\{T_0,T_1,T_2,T_3,T_4,T_5,T_6\}=\mcl{M}_{\epsilon_1}\left(M,K_1,K_2\right)$. Then, for any $w\in \lt$, if the operator $\pop$ is given by
\begin{equation}\label{eqn:pop}
\left(\pop w \right)(x)=M(x)w(x)+\igzo K(x,\xi)w(\xi)dx,
\end{equation} with
\[
K(x,\xi) = \begin{cases} K_1(x,\xi) & \xi \leq x \\
K_2(x,\xi) &  \xi>x \end{cases}
\] and operator $\mcl{A}$ is given by Equation~\eqref{eqn:A_prelim}, we have that
\begin{align*}
&\ip{\mcl{AP}z}{z}+\ip{z}{\mcl{AP}z} \\
&\leq \ip{\mcl{T}z}{z}+z(0)\left(T_3 z(0)+\igzo T_4(x)z(x)dx \right) \\
& \quad \quad +z(1) \left(T_5 z(1)+T_6 z_x(1) \right),
\end{align*}
where $z=\pinv w$ for any $w \in H^2(0,1)$ with $w_x(0)=0$. Here we define the operator $\mcl{T}$ as
\begin{align*}
\left(\mcl{T}z \right)(x)=T_0(x)z(x)&+\igzo T(x,\xi)z(\xi)d\xi,
\end{align*} with
\[T(x,\xi) = \begin{cases} T_1(x,\xi) & \xi \leq x \\ 
T_2(x,\xi) &  \xi>x \end{cases}.\]
\end{lemma}
The second lemma allows us to represent the derivative of $\ip{\mcl{A}w}{\sop w}+\ip{\mcl{SA}w}{w}$.
\begin{lemma}\label{lem:primal_LOI}
For any $\{N,P_1,P_2\} \in \Xi_{\{d_1,d_2,\epsilon_2,\epsilon_2\}}$, $0<\epsilon_1<\epsilon_2<\infty$, let $\{Q_0,Q_1,Q_2,Q_3,Q_4,Q_5,Q_6,Q_7,Q_8\}=\mcl{N}_{\epsilon_1}\left(N,P_1,P_2\right)$. Then, for any $z \in \lt$, if operator $\sop$ is given by
\begin{equation}\label{eqn:sop}
\left(\sop z \right)(x)=N(x)z(x)+\igzo P(x,\xi)z(\xi)dx,
\end{equation}
 with
\[P(x,\xi) = \begin{cases} P_1(x,\xi) & \xi \leq x \\
P_2(x,\xi) &  \xi>x \end{cases} ,\]
and operator $\mcl{A}$ is given by Equation~\eqref{eqn:A_prelim},
we have that
\begin{align*}
&\ip{\mcl{A}w}{\sop w}+\ip{\mcl{SA}w}{w}\\
&\leq  \ip{\mcl{Q}w}{w}+w(0)\left(Q_3 w(0) + \igzo Q_4(x)w(x)dx \right) \\
& \quad \quad +w(1) \left(Q_5w(1)+\igzo Q_6(x)w(x)dx \right) \\
& \quad \quad + w_x(1)\left(Q_7w(1)+ \igzo Q_8(x)w(x)dx\right),
\end{align*} for any $w \in H^2(0,1)$ with $w_x(0)=0$. Here we define the operator $\mcl{Q}$, for any $z \in \lt$, as
\begin{align*}
\left(\mcl{Q}z \right)(x)=Q_0(x)z(x)&+\igzo Q(x,\xi)z(\xi)d \xi,
\end{align*} with
\[Q(x,\xi) = \begin{cases} Q_1(x,\xi) & \xi \leq x \\
Q_2(x,\xi) &  \xi>x \end{cases}.\]
\end{lemma}
\section{OUTPUT FEEDBACK CONTROLLER SYNTHESIS}\label{sec:control}

Our approach to design of output-feedback controllers is based on three steps. First, we design the control operator $\mcl{F}$ which maps that state to the control input as $u(t)=\mcl{F}w$. However, because we cannot measure the state, we find function $O_1(x)$ and scalar $O_2$ which define the observer which outputs an estimate of the state $\hat w$. Finally, we prove that the controller coupled to the observer as $u(t)=\mcl{F}\hat w$ produces a closed-loop system with bounded $L_2$ gain from exogenous input to controlled output.

\subsection{Control Design}
We begin by designing the control operator $\mcl{F} \in \mcl{L}\left(H^2(0,1),\R \right)$. Consider the following observer dynamics

\begin{align}
\wh_t(x,t)=&a(x)\wh_{xx}(x,t)+b(x)\wh_x(x,t)+c(x)\wh(x,t) \nonumber \\
&\label{eqn:obserror_PDE} +O_1(x)e_x(1,t),
\end{align}
with boundary conditions
\begin{equation}\label{eqn:obserror_PDE_BC}
\wh_x(0,t)=0, \quad \wh(1,t)=u(t)+O_2 e_x(1,t).
\end{equation} The following lemma defines the operator $\mcl{F}$.
\begin{lemma}\label{lem:control}
Suppose there exist $\{M,K_1,K_2\} \in \Xi_{d_1,d_2,\epsilon_1,\epsilon_2}$ and $T_i$ such that
$\{T_0,T_1,T_2,T_3,T_4,T_5,T_6\}=\mcl{M}_{\epsilon_1}\left(M,K_1,K_2\right)$ and
\[
 T_3 \leq 0, \quad  T_4(x)=0.
\]
Let $u(t)=\mcl{F}\wh(\cdot,t)$ where $\mcl{F}=\mcl{Z}\pinv$, $\pop$ is as in Eqn.~\eqref{eqn:pop} and the operator $\mcl{Z}$ is defined as
\[
\left(\mcl{Z}g \right)(x):=Z_1 g_x(1)+\igzo K_1(1,x)g(x)dx.
\]
where $Z_1$ is any scalar such that $T_5 Z_1> -T_6 M(1)$.
Now, if $V(\wh(\cdot,t))=\ip{\wh(\cdot,t)}{\pinv \wh(\cdot,t)}$ where $\pop$ is as in Equation~\eqref{eqn:pop}. Then for any solution $\wh$ of Eqns.~\eqref{eqn:obserror_PDE} and~\eqref{eqn:obserror_PDE_BC} with input $e_x$ , we have that
\begin{align*}
\frac{d}{dt}V(\wh(\cdot,t)) \leq & \ip{\mcl{T}\zh(\cdot,t)}{\zh(\cdot,t)}+2 \ip{O_1(\cdot)e_x(1,t)}{\zh(\cdot,t)} \\
& -\mu \zh(1,t)^2 -\frac{T_6 O_2}{Z_1}\zh(1,t)e_x(1,t),
\end{align*}
for some $\mu>0$ where $\hat z=\pinv \hat w$ and $\mcl{T}$ is defined in Lemma~\ref{lem:dual_LOI}.
\end{lemma}

%
\begin{proof}
We begin by taking the time derivative of the Lyapunov function $V_o(\wh(\cdot,t))$ along the trajectories of~\eqref{eqn:obserror_PDE}-\eqref{eqn:obserror_PDE_BC}
\begin{align*}
\frac{d}{dt}V_o(\wh(\cdot,t))=&\ip{\wh_t(\cdot,t)}{\pinv \wh(\cdot,t)}+\ip{\pinv \wh(\cdot,t)}{ \wh_t(\cdot,t)}\\
=&\ip{\mcl{A}\wh(\cdot,t)}{\pinv \wh(\cdot,t)}+\ip{\pinv \wh(\cdot,t)}{ \mcl{A}\wh(\cdot,t)}\\
&+2 \ip{O_1(\cdot,t)e_x(1,t)}{\pinv \wh(\cdot,t)},
\end{align*} where we use that $\pinv$ is self-adjoint have simplified the derivative using the definition of operator $\mcl{A}$ provided in Equation~\eqref{eqn:A_prelim}. We rewrite this as
\begin{align*}
&\frac{d}{dt}V_o(\wh(\cdot,t))\\
&=\ip{\mcl{A}\pop \pinv \wh(\cdot,t)}{\pinv \wh(\cdot,t)}+\ip{\pinv \wh(\cdot,t)}{ \mcl{A} \pop \pinv \wh(\cdot,t)}\\
&\quad \quad +2 \ip{O_1(\cdot,t)e_x(1,t)}{\pinv \wh(\cdot,t)}.
\end{align*} Now define $\zh=\pinv \wh$, then
\begin{align*}
\frac{d}{dt}V_o(\wh(\cdot,t))=&\ip{\mcl{A}\pop \zh(\cdot,t)}{\zh(\cdot,t)}+\ip{\zh(\cdot,t)}{ \mcl{A} \pop \zh(\cdot,t)}\\
& +2 \ip{O_1(\cdot,t)e_x(1,t)}{\zh(\cdot,t)}.
\end{align*} Now, applying Lemma~\ref{lem:dual_LOI} and using the facts that $T_3 \leq 0$ and $T_4(x)=0$ produces
\begin{align}
\frac{d}{dt}V_o(\wh(\cdot,t)) \leq & \ip{\mcl{T}\zh(\cdot,t)}{\zh(\cdot,t)}+2 \ip{O_1(\cdot,t)e_x(1,t)}{\zh(\cdot,t)} \nonumber \\
&\label{eqn:cont:Vdot1} +\zh(1,t) \left(T_5 \zh(1,t)+T_6 \zh_x(1,t) \right).
\end{align} Since $\zh=\pinv \wh$, $\wh=\pop \zh$. Thus
\begin{equation}\label{eqn:cont:1}
\wh(1,t)=M(1)\zh(1,t)+\igzo K_1(1,x)\zh(x,t)dx.
\end{equation} From the boundary condition in~\eqref{eqn:obserror_PDE_BC} we get
\begin{align*}
\wh(1,t)=u(t)+O_2 e_x(1,t)=& \mcl{F}\wh(\cdot,t)+O_2 e_x(1,t)\\
=& \mcl{F}\pop \pinv \wh(\cdot,t)+O_2 e_x(1,t)\\
=& \mcl{Z}\zh(\cdot,t)+O_2 e_x(1,t).
\end{align*} Using the definition of $\mcl{Z}$,
\[\wh(1,t)=Z_1 \zh_x(1,t)+\igzo Z_2(x)\zh(x,t)dx+O_2 e_x(1,t).\] Substituting into Equation~\eqref{eqn:cont:1} and using the definition $Z_2(x)=K_1(1,x)$ we get
\[\zh_x(1,t)=\frac{M(1)}{Z_1}\zh(1,t)-\frac{O_2}{Z_1}e_x(1,t).\] Substituting this expression into~\eqref{eqn:cont:Vdot1}
\begin{align*}
\frac{d}{dt}V_o(\wh(\cdot,t)) \leq & \ip{\mcl{T}\zh(\cdot,t)}{\zh(\cdot,t)}+2 \ip{O_1(\cdot,t)e_x(1,t)}{\zh(\cdot,t)} \\
&+\left(T_5+\frac{T_6 M(1)}{Z_1} \right)\zh(1,t)^2\\
& -\frac{T_6 O_2}{Z_1}\zh(1,t)e_x(1,t).
\end{align*} Now, since $Z_1<0$ is a scalar such that $T_5+T_6M(1)/Z_1<0$, there exists a scalar $\mu>0$ such that $T_5+T_6 M(1) /Z_1=-\mu$. Hence
\begin{align*}
\frac{d}{dt}V_o(\wh(\cdot,t)) \leq & \ip{\mcl{T}\zh(\cdot,t)}{\zh(\cdot,t)}+2 \ip{O_1(\cdot)e_x(1,t)}{\zh(\cdot,t)} \\
& -\mu \zh(1,t)^2 -\frac{T_6 O_2}{Z_1}\zh(1,t)e_x(1,t).
\end{align*}
\end{proof}
\subsection{Observer Design}

We now design the function $O_1(x)$ and scalar $O_2$ which define the observer.
We begin by subtracting Equations~\eqref{eqn:prob:PDE_form}-\eqref{eqn:prob:PDE_form_BC} from~\eqref{eqn:prob:observer_form}-\eqref{eqn:prob:observer_form_BC} to obtain the dynamics of the error variable $e=\wh-w$ given by
\begin{align}
e_t(x,t)=&a(x)e_{xx}(x,t)+b(x)e_x(x,t)+c(x)e(x,t)\nonumber \\
&\label{eqn:error_PDE}+O_1(x)e_x(1,t)-f(x,t),
\end{align} with boundary conditions
\begin{equation}\label{eqn:error_PDE_BC}
e_x(0,t)=0, \quad e(1,t)=O_2 e_x(1,t),
\end{equation} where we have used the definition of the measurement $y(t)=w_x(1,t)$ and $\hat{y}(t)=\wh_x(1,t)$. We present the following lemma.
\begin{lemma}\label{lem:observer}
Suppose there exist $\{N,P_1,P_2\} \in \Xi_{d_1,d_2,\epsilon_1,\epsilon_2}$, such that
\[
Q_3 \leq 0, \quad  Q_4(x)=0,
\]
where $\{Q_0,Q_1,Q_2,Q_3,Q_4,Q_5,Q_6,Q_7,Q_8\}=\mcl{N}_{\epsilon_1}\left(N,P_1,P_2\right)$ (See Defn.~\ref{def:primal}).
Let $\sop$ be defined as in Eqn.~\eqref{eqn:sop}. Then choose a scalar $O_2<0$ such that
\[
Q_5+\frac{1}{O_2}Q_7 <0,
\] and let $O_1(x)=\left(\sinv R_1\right)(x)$ where $R_1(x)=-\frac{1}{2} \left(O_2 Q_6(x)+Q_8(x) \right)$.
Define $V_e(e) = \ip{e}{\sop e}$. Then for any $e,f$ which satisfies Eqns.~\eqref{eqn:error_PDE}-\eqref{eqn:error_PDE_BC} with $O_1(x)$ and $O_2$ as defined here, we have
\begin{align*}
\frac{d}{dt}V_e(e(\cdot,t))\leq &\ip{\mcl{Q}e(\cdot,t)}{e(\cdot,t)}+2\ip{f(\cdot,t)}{\sop e(\cdot,t)} \nonumber \\
&  -\zeta e(1,t)^2,
\end{align*}
for some scalar $\zeta>0$ where $\mcl{Q}$ is as defined in Lemma~\ref{lem:primal_LOI}.
\end{lemma}

%
%
\begin{proof}
We begin by taking the time derivative of the Lyapunov function $V_e(e(\cdot,t))$ along the trajectories of~~\eqref{eqn:error_PDE}-\eqref{eqn:error_PDE_BC}, yielding
\begin{align}
&\frac{d}{dt}V_e(e(\cdot,t))=\ip{e_t(\cdot,t)}{\sop e(\cdot,t)}+\ip{e(\cdot,t)}{\sop e_t(\cdot,t)} \nonumber \\
&\quad = \ip{\mcl{A}e(\cdot,t))}{\sop e(\cdot,t)}+\ip{e(\cdot,t))}{\sop \mcl{A} e(\cdot,t)} \nonumber \\
&\label{eqn:obs:Vdot_1}\qquad + 2 \ip{e(\cdot,t)}{\left(\sop O_1 \right)(\cdot)e_x(1,t)}+2\ip{f(\cdot,t)}{\sop e(\cdot,t)},
\end{align} where we have again used the definition of $\mcl{A}$ from Eqn.~\eqref{eqn:A_prelim} and we have also usd the fact that $\sop$ is self-adjoint. Now, since from the theorem statement we have that $Q_3 \leq 0$ and $Q_4(x)=0$, applying Lemma~\ref{lem:primal_LOI} produces
\begin{align}
\frac{d}{dt}V_e(e(\cdot,t))&\leq \ip{\mcl{Q}e(\cdot,t)}{e(\cdot,t)}+2\ip{f(\cdot,t)}{\sop e(\cdot,t)} \nonumber \\
&+ e(1,t)\left(Q_5 e(1,t)+\igzo Q_6(x)e(x,t)dx \right) \nonumber \\
&+ e_x(1,t)\left(Q_7 e(1,t)+\igzo Q_8(x)e(x,t)dx \right) \nonumber \\
&\label{eqn:obs:Vdot_2}+ 2 \ip{e(\cdot,t)}{R_1(\cdot)e_x(1,t)},
\end{align} where we have used the fact that since $O_1(x)=\left(\sinv R_1\right)(x)$, $R_1(x)=\left(\sop O_1\right)(x)$. We have the boundary condition $e(1,t)=O_2 e_x(1,t)$ and since $O_2<0$, we have that $e_x(1,t)=e(1,t)/O_2$. Substituting in~\eqref{eqn:obs:Vdot_2},
\begin{align}
&\frac{d}{dt}V_e(e(\cdot,t)) \nonumber \\
&\leq \ip{\mcl{Q}e(\cdot,t)}{e(\cdot,t)}+2\ip{f(\cdot,t)}{\sop e(\cdot,t)} \nonumber \\
& \quad +\left(Q_5+\frac{1}{O_2}Q_7 \right)e(1,t)^2 \nonumber \\
&\label{eqn:obs:Vdot_3} \quad +e(1,t)\igzo \left(Q_6(x)+\frac{1}{O_2}Q_8(x)+\frac{2}{O_2} R_1(x) \right)e(x,t)dx.
\end{align} Since $O_2<0$ is a scalar such that $Q_5+Q_7/O_2<0$, let
\begin{equation}\label{eqn:obs:gain1}
\zeta=-(Q_5+\frac{1}{O_2}Q_7).
\end{equation} Then $\zeta>0$. Now, using the definition of $R_1(x)$ we get that
\begin{equation}\label{eqn:obs:gain2}
Q_6(x)+\frac{1}{O_2}Q_8(x)+\frac{2}{O_2} R_1(x)=0.
\end{equation} Substituting Eqns.~\eqref{eqn:obs:gain1}-\eqref{eqn:obs:gain2} into Eqn.~\eqref{eqn:obs:Vdot_3}, we find
\begin{align*}
\frac{d}{dt}V_e(e(\cdot,t))\leq&\ip{\mcl{Q}e(\cdot,t)}{e(\cdot,t)}+2\ip{f(\cdot,t)}{\sop e(\cdot,t)} \nonumber \\
&  -\zeta e(1,t)^2.
\end{align*}
\end{proof}

\subsection{Output Feedback Based Control}
We now have the following set of coupled parabolic PDEs.
\begin{align}
w_t(x,t)=&\label{eqn:couple1}a(x)w_{xx}(x,t)+b(x)w_x(x,t)+c(x)w(x,t)+f(x,t),\\
\wh_t(x,t)=&a(x)\wh_{xx}(x,t)+b(x)\wh_x(x,t)+c(x)\wh(x,t) \nonumber \\
&\label{eqn:couple2}+O_1(x) \left(\wh_x(1,t)-w_x(1,t) \right),
\end{align}
with boundary conditions
\begin{align}
w_x(0,t)=0, \quad &\label{eqn:couple3}w(1,t)=\mcl{F}\wh(\cdot,t),\\
\wh_x(0,t)=0, \quad &\label{eqn:couple4}\wh(1,t)=\mcl{F}\wh(\cdot,t)+O_2\left(\wh_x(1,t)-w_x(1,t) \right).
\end{align}

We now prove that the previously designed controller and the observer can be coupled such that norm of the system state remains bounded in the presence of an exogenous input.
\begin{theorem}\label{thm:coupled}
Suppose there exist scalars $0<\epsilon_1<\epsilon_2<\infty$, $\delta,\beta>0$ $d_1,d_2 \in \N$ and polynomials $\{N,P_1,P_2\} \in \Xi_{d_1,d_2,\epsilon_1,\epsilon_2}$ and $\{N,P_1,P_2\} \in \Xi_{d_1,d_2,\epsilon_1,\epsilon_2}$, such that
\begin{align*}
&\{-T_0-2\delta M,-T_1-2\delta K_1,-T_2-2\delta K_2\} \in \Xi_{d_1,d_2,0,\beta},\\
&\{-Q_0-2\delta N,-Q_1-2\delta P_1,-Q_2-2\delta P_2\} \in \Xi_{d_1,d_2,0,\beta},\\
&T_3 \leq 0, \quad T_4(x)=0, \quad Q_3 \leq 0, \quad Q_4(x)=0,
\end{align*} where $\{T_0,T_1,T_2,T_3,T_4,T_5,T_6\}=\mcl{M}_{\epsilon_1}\left(M,K_1,K_2\right)$ and  $\{Q_0,Q_1,Q_2,Q_3,Q_4,Q_5,Q_6,Q_7,Q_8\}=\mcl{N}_{\epsilon_1}\left(N,P_1,P_2\right)$ as provided in Definitions~\ref{def:dual} and~\ref{def:primal} respectively. Let $\pop$ be defined as in Equation~\eqref{eqn:pop} and $\sop$ as in Equation~\eqref{eqn:sop}.

Then for any solution $\bmat{w(x,t) & \wh(x,t)}$ of the coupled dynamics~\eqref{eqn:couple1}-\eqref{eqn:couple4}, if $\mcl{F}$ is given by Lemma~\ref{lem:control} and  $O_1(x)$ and $O_2$ are given by Lemma~\ref{lem:observer}, there exists a scalar $\gamma>0$ such that
\[\norm{w}_{L_2\left(0,\infty;\lt \right)} \leq \gamma \norm{f}_{L_2\left(0,\infty;\lt \right)},\] for any $f \in L_2\left(0,\infty;\lt \right)$.
\end{theorem}
\begin{proof}
For the Lyapunov function $V_o(\wh(\cdot,t))=\ip{\wh(\cdot,t)}{\pinv w(\cdot,t)}$, we have from Lemma~\ref{lem:control} that there exists  scalar $\mu>0$ such that
\begin{align*}
\frac{d}{dt}V_o(\wh(\cdot,t)) \leq & \ip{\mcl{T}\zh(\cdot,t)}{\zh(\cdot,t)}+2 \ip{O_1(\cdot)e_x(1,t)}{\zh(\cdot,t)} \\
& -\mu \zh(1,t)^2 -\frac{T_6 O_2}{Z_1}\zh(1,t)e_x(1,t).
\end{align*} We have from Lemma~\ref{lem:observer} that $e_x(1,t)=e(1,t)/O_2$. Therefore
\begin{align}
\frac{d}{dt}V_o(\wh(\cdot,t)) \leq & \ip{\mcl{T}\zh(\cdot,t)}{\zh(\cdot,t)}+\frac{2}{O_2} \ip{O_1(\cdot)e(1,t)}{\zh(\cdot,t)} \nonumber \\
&\label{eqn:coupled:Vobs} -\mu \zh(1,t)^2 -\frac{T_6 }{Z_1}\zh(1,t)e(1,t).
\end{align}

For the Lyapunov function $V_e(e(\cdot,t))=\ip{e(\cdot,t)}{\sop e(\cdot,t)}$, we have from Lemma~\ref{lem:observer} that there exists a scalar $\zeta>0$ such that
\begin{align}
\frac{d}{dt}V_e(e(\cdot,t))\leq &\ip{\mcl{Q}e(\cdot,t)}{e(\cdot,t)}+2\ip{f(\cdot,t)}{\sop e(\cdot,t)} \nonumber \\
&\label{eqn:coupled:Vcont}  -\zeta e(1,t)^2.
\end{align} From Equations~\eqref{eqn:coupled:Vobs}-\eqref{eqn:coupled:Vcont} we conclude that for any $A>0$
\begin{align}
&\frac{d}{dt}V_o(\wh(\cdot,t))+A \frac{d}{dt}V_e(e(\cdot,t)) \nonumber \\
&\leq A\ip{\mcl{Q}e(\cdot,t)}{e(\cdot,t)}+2A \ip{f(\cdot,t)}{\sop e(\cdot,t)} \nonumber  \\
&\label{eqn:coupled:Vtotal_1} \quad + \left\langle \bmat{\zh(\cdot,t) \\ \zh(1,t) \\ e(1,t)}, \bmat{\mcl{T} & 0 & \mcl{O} \\ \star & -\mu & -\frac{T_6}{2Z_1} \\ \star & \star & -A \zeta} \bmat{\zh(\cdot,t) \\ \zh(1,t) \\ e(1,t)}  \right\rangle,
\end{align} where $\left(\mcl{O}y \right)(x)=\frac{1}{O_2} O_1(x)y(x)$, for any $y \in \lt$, and the inner product is defined on $\lt \times \lt \times \lt$. Now, since $\{-T_0-2\delta M,-T_1-2\delta K_1,-T_2-2\delta K_2\} \in \Xi_{d_1,d_2,0,\beta}$, we have that $\mcl{T}+2\delta \pop \leq 0$. Therefore, for any $0<\theta<\delta$, it can be established using Schur complement that for a large enough $A>0$,
\[\bmat{\mcl{T}+2\theta \pop & 0 & \mcl{O} \\ \star & -\mu & -\frac{T_6}{2Z_1} \\ \star & \star & -A \zeta} \leq 0.\] Therefore
\[\bmat{\mcl{T}  & 0 & \mcl{O} \\ \star & -\mu & -\frac{T_6}{2Z_1} \\ \star & \star & -A \zeta} \leq \bmat{-2 \theta \pop & 0 & 0\\ \star & 0 & 0\\ \star & \star & 0}.\] Substituting into Equation~\eqref{eqn:coupled:Vtotal_1}, we get
\begin{align}
&\frac{d}{dt}V_o(\wh(\cdot,t))+A \frac{d}{dt}V_e(e(\cdot,t)) \nonumber \\
&\leq A\ip{\mcl{Q}e(\cdot,t)}{e(\cdot,t)}+2A \ip{f(\cdot,t)}{\sop e(\cdot,t)} \nonumber  \\
& \quad - 2 \theta \ip{\zh(\cdot,t)}{\pop \zh(\cdot,t)}. \nonumber
\end{align} Let $V_o(\wh(\cdot,t))+A V_e(e(\cdot,t))=V(t)$, thus
\begin{align*}
\frac{d}{dt}V(t)+2\theta \ip{\zh(\cdot,t)}{\pop \zh(\cdot,t)} \leq & A\ip{\mcl{Q}e(\cdot,t)}{e(\cdot,t)}\\
&+2A \ip{f(\cdot,t)}{\sop e(\cdot,t)}.
\end{align*} Adding $A \delta \ip{\sop e(\cdot,t))}{e(\cdot,t)}-\frac{A}{\delta}\ip{f(\cdot,t)}{\sop f(\cdot,t)}$ to both sides,
\begin{align}
&\frac{d}{dt}V(t)+2\theta \ip{\zh(\cdot,t)}{\pop \zh(\cdot,t)}+A\delta \ip{\sop e(\cdot,t))}{e(\cdot,t)}\\
&-\frac{A}{\delta}\ip{f(\cdot,t)}{\sop f(\cdot,t)} \\
&\label{eqn:coupled:Vtotal_2}\leq  A \left\langle \bmat{e(\cdot,t) \\ f(\cdot,t)}, \bmat{\mcl{Q}+\delta \sop & \sop \\ \star & -\frac{1}{\delta}\sop}\bmat{e(\cdot,t) \\ f(\cdot,t)} \right\rangle.
\end{align} Since $\{-Q_0-2\delta N,-Q_1-2\delta P_1,-Q_2-2\delta P_2\} \in \Xi_{d_1,d_2,0,\beta}$, we have that $\mcl{Q}+2\delta \sop \leq 0$. Hence, using Schur complement we conclude that
\[\bmat{\mcl{Q}+\delta \sop & \sop \\ \star & -\frac{1}{\delta}\sop} \leq 0.\]
Therefore, from Equation~\eqref{eqn:coupled:Vtotal_2} we conclude that
\begin{align}
&\frac{d}{dt}V(t)+2\theta \ip{\zh(\cdot,t)}{\pop \zh(\cdot,t)}+A\delta \ip{\sop e(\cdot,t))}{e(\cdot,t)} \nonumber \\
&\label{eqn:couple:Vtotal_3} \leq \frac{A}{\delta}\ip{f(\cdot,t)}{\sop f(\cdot,t)}.
\end{align} Since the operator $\sop$ is defined using $\{N,P_1,P_2\} \in \Xi_{d_1,d_2,\epsilon_1,\epsilon_2}$, we have from Theorem~\ref{thm:jointpos} that, for all $g \in \lt$,
\begin{equation}\label{eqn:coupled:bound1}
\epsilon_1 \norm{g}^2 \leq \ip{\sop g}{g} \leq \epsilon_2 \norm{g}^2.
\end{equation} Similarly, since $\pop$ is defined using $\{M,K_1,K_2\} \in \Xi_{d_1,d_2,\epsilon_1,\epsilon_2}$, using Theorem~\ref{thm:jointpos} it can established that
\[\frac{1}{\epsilon_2} \norm{g}^2 \leq \ip{\pinv g}{g} \leq \frac{1}{\epsilon_1} \norm{g}^2.\] Since $\ip{\zh(\cdot,t)}{\pop \zh(\cdot,t)} = \ip{\pinv \wh(\cdot,t)}{ \wh(\cdot,t)}$, from the previous expression we have that
\begin{equation}\label{eqn:coupled:bound2}
\frac{1}{\epsilon_2} \norm{\wh(\cdot,t)}^2 \leq \ip{\pinv \wh(\cdot,t)}{ \wh(\cdot,t)}=\ip{\zh(\cdot,t)}{\pop \zh(\cdot,t)}.
\end{equation} Substituting Equation~\eqref{eqn:coupled:bound2} in Equation~\eqref{eqn:couple:Vtotal_3} and using~\eqref{eqn:coupled:bound1}, we get
\begin{align*}
&\frac{d}{dt}V(t)+2\frac{\theta}{\epsilon_2} \norm{\wh(\cdot,t)}^2+A\delta \epsilon_1 \norm{e(\cdot,t)}^2  \leq \frac{A \epsilon_2 }{\delta}\norm{f(\cdot,t)}^2.
\end{align*} Integrating in time from $0$ to some $0<T<\infty$, we get
\begin{align}
&V(T)-V(0)+2\frac{\theta}{\epsilon_2} \int_0^T \norm{\wh(\cdot,t)}^2dt+A\delta \epsilon_1 \int_0^T \norm{e(\cdot,t)}^2 dt \nonumber  \\
&\label{eqn:couple:Vtotal_4}\leq \frac{A \epsilon_2 }{\delta}\int_0^T \norm{f(\cdot,t)}^2 dt.
\end{align} Now,
$V(T)=V_o(\wh(\cdot,T))+A V_e(e(\cdot,T)) \geq 0$. Additionally, if we assume zero initial conditions, then
$V(0)=V_o(\wh(\cdot,0))+A V_e(e(\cdot,0)) = 0$. Therefore we conclude from Equation~\eqref{eqn:couple:Vtotal_4} that
\[ \int_0^T \norm{\wh(\cdot,t)}^2dt+\int_0^T \norm{e(\cdot,t)}^2 dt \nonumber \leq \frac{A \epsilon_2 }{\delta \nu}\int_0^T \norm{f(\cdot,t)}^2 dt,\] where $\nu=\min \left\lbrace 2\frac{\theta}{\epsilon_2},A\delta \epsilon_1 \right\rbrace$. Since $f \in L_2(0,\infty;\lt)$, taking the limit $T \rightarrow \infty$, we get
\begin{align*}
& \norm{\wh}^2_{L_2(0,\infty;\lt)}+ \norm{e}^2_{L_2(0,\infty;\lt)} \nonumber \\
&\leq \frac{A \epsilon_2 }{\delta \nu}\norm{f}^2_{L_2(0,\infty;\lt)}.
\end{align*} Hence, we conclude that
\begin{align*}
\norm{\wh}_{L_2(0,\infty;\lt)} \leq &\sqrt{\frac{A \epsilon_2 }{\delta \nu}}\norm{f}_{L_2(0,\infty;\lt)},\\
\norm{e}_{L_2(0,\infty;\lt)} \leq &\sqrt{\frac{A \epsilon_2 }{\delta \nu}}\norm{f}_{L_2(0,\infty;\lt)}.
\end{align*} Since $e=\wh-w$, $w=\wh-e$. Therefore
\begin{align*}
\norm{w}_{L_2(0,\infty;\lt)} &\leq \norm{\wh}_{L_2(0,\infty;\lt)}+\norm{e}_{L_2(0,\infty;\lt)}\\
&\leq 2 \sqrt{\frac{A \epsilon_2 }{\delta \nu}}\norm{f}_{L_2(0,\infty;\lt)}.
\end{align*} Setting $\gamma=2\sqrt{\frac{A \epsilon_2 }{\delta \nu}}$ completes the proof.
\end{proof}
\section{NUMERICAL RESULTS}\label{sec:num_results}
In this section we consider a couple of examples on which we test the conditions of Theorem~\ref{thm:coupled} using SOS and SDP. These numerical results are obtained using the Matlab toolbox SOSTOOLS~\cite{prajna2001introducing}.

We consider the following two PDEs. First consider the classical heat equation with an unsteady source term.
\begin{equation}\label{eqn:exmp1}
w_t(x,t)=w_{xx}(x,t)+\lambda w(x,t)+f(x,t),
\end{equation}
Without feedback, this system is unstable for $\lambda>\pi^2/4$. Next, we consider a randomly generated PDE.
\begin{align}
w_t(x,t)&= \left(x^3-x^2+2 \right)w_{xx}(x,t)+\left(3x^2-2x \right)w_x(x,t) \nonumber \\
&\label{eqn:exmp2}+ \left( -0.5x^3+1.3x^2-1.5x+0.7x+\lambda \right)w(x,t)+f(x,t)
\end{align} with boundary conditions
\begin{equation}
w_x(0,t)=0, \quad w(1,t)=u(t),
\end{equation} By using stability analysis and numerical simulation, we estimate that PDE is unstable for $\lambda>4.66$.

In these examples, we find the maximum $\lambda$, using a bisection search, for which we can construct stabilizing output-based boundary feedback controllers. We test the conditions of Theorem~\ref{thm:coupled} with $\epsilon_1=0.001$, $\epsilon_2=1$, $\delta=0.001$ and increasing values of $d_1$ and $d_2$. Table~\ref{table:exmp1:cont} presents the maximum $\lambda>0$  for which we can construct output feedback controllers for PDE~\eqref{eqn:exmp1} as a function of the degree $d_1=d_2=d$ of the polynomials which define the controller, observer, and Lyapunov function. Table~\ref{table:exmp2:cont} presents the maximum $\lambda>0$ for PDE~\eqref{eqn:exmp2}.

\begin{table}{}
\caption{Maximum $\lambda$ as a function of polynomial degree $d_1=d_2=d$ for which we can construct output feedback boundary controllers for PDE~\eqref{eqn:exmp1}.}
\vspace{-10pt}
\begin{center}
    \begin{tabular}{l *{7}{c}}\hline
  $d=7$ & $8$ & $9$ & $10$ & $11$  \\ \hline
  $\lambda=12.69$ & $16.01$ & $17.96$ & $17.96$ & $21.97$
\end{tabular}
\end{center}
\label{table:exmp1:cont}
\end{table}

\begin{table}{}
\caption{Maximum $\lambda$ as a function of polynomial degree $d_1=d_2=d$ for which we can construct output feedback boundary controllers for PDE~\eqref{eqn:exmp2}.}
\vspace{-10pt}
\begin{center}
    \begin{tabular}{l *{7}{c}}\hline
  $d=4$ & $5$ & $6$ & $7$ & $8$  \\ \hline
  $\lambda=18.75$ & $28.78$ & $32.03$ & $32.03$ & $39.16$
\end{tabular}
\end{center}
\label{table:exmp2:cont}
\end{table}

The numerical results suggest that increasing the degree $d_1=d_2=d$ of the polynomial representation leads to higher values of $\lambda>0$. Moreover, the value of $\lambda>0$ does not appear to be upper bounded, which implies that the method is asymptotically accurate. That is given any $\lambda>0$, we conjecture that we can construct output feedback controllers for a large enough degree $d$.

Figures~\ref{fig:1}-\ref{fig:2} represent the simulation of PDE~\eqref{eqn:exmp2} with $\lambda=39$ subject to the output feedback based control in the presence of exogenous input $f(x,t)=e^{-t}\cos(\pi t)\left(1+\sin(0.1 \pi x)\right)$.
\begin{figure}[h!]
\vspace{-10pt}
    \centering
    \includegraphics[width=0.4\textwidth]{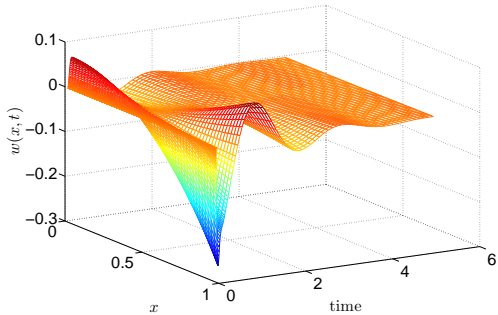}
    \vspace{-5pt}
\caption{State of PDE~\eqref{eqn:exmp2} with point observation and point actuation.}
\label{fig:1}
\end{figure}
\begin{figure}[h!]
\vspace{-10pt}
    \centering
    \includegraphics[width=0.4\textwidth]{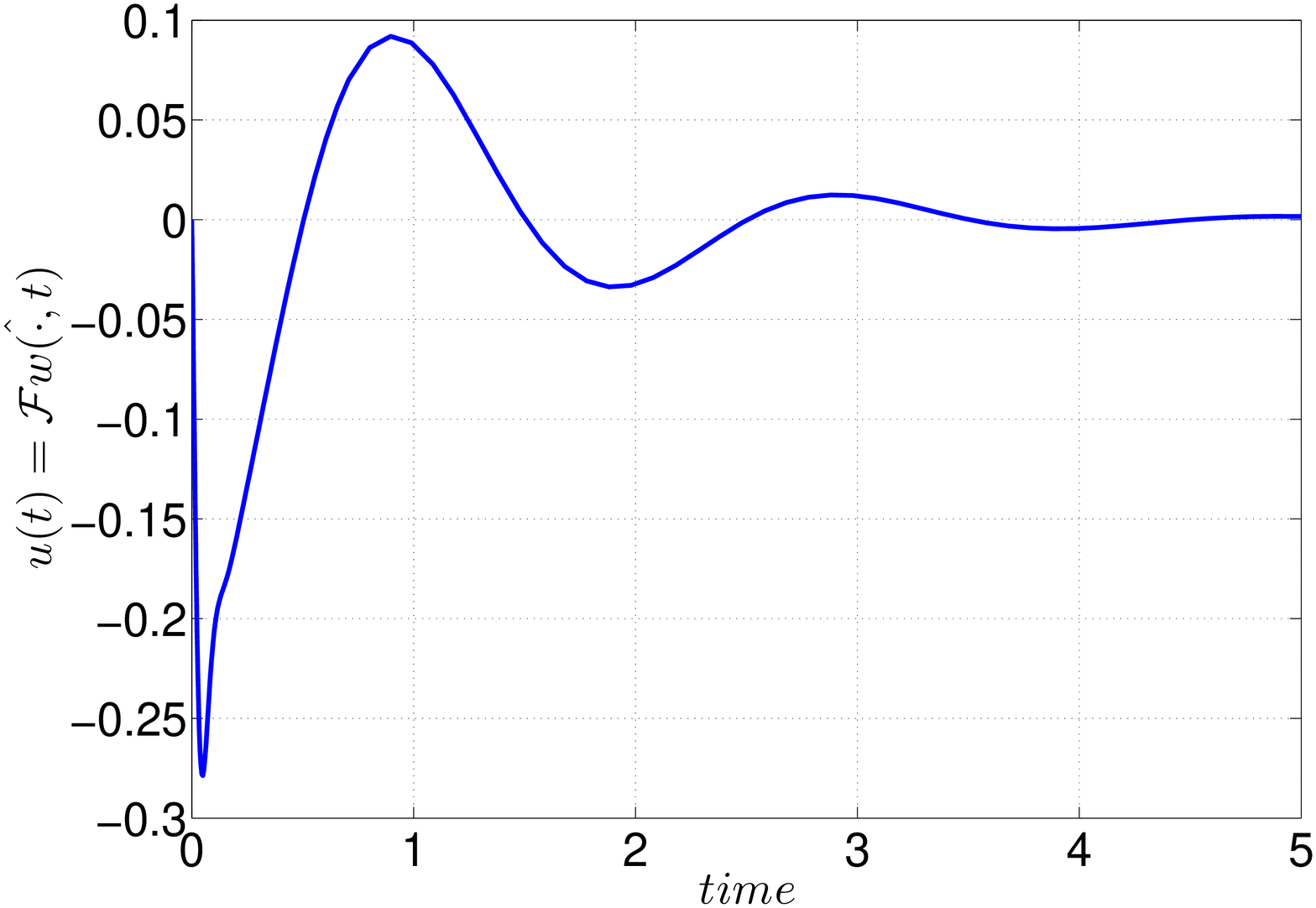}
    \vspace{-5pt}
\caption{Boundary control effort $w(1,t)=u(t)$ for PDE~\eqref{eqn:exmp2}.}
\label{fig:2}
\end{figure}

One of the key technical advances of this paper is the use of semi-separable kernels $K_1$, $K_2$, $P_1$ and $P_2$ and this advance leads to significantly more complex stability conditions. Therefore we wish to establish if the inclusion of the variables $K_1$, $K_2$, $P_1$ and $P_2$ does, in fact, provide any significant performance gain. In order to do this, we check the conditions of Theorem~\ref{thm:coupled} while setting $K_1=K_2=P_1=P_2=0$ (similar to our previous approach~\cite{gahlawat2011designing}) and applied these conditions to the example PDEs.  Table~\ref{table:simple_lyap1:cont1} presents these results for PDE~\eqref{eqn:exmp1} and Table~\ref{table:simple_lyap1:cont2} presents results for PDE~\eqref{eqn:exmp2}.

\begin{table}{}
\caption{Maximum $\lambda$ as a function of polynomial degree, $d_1=d_2=d$ for PDE~\eqref{eqn:exmp1} for which we can construct controllers using  with $K_1=K_2=P_1=P_2=0$.}
\vspace{-10pt}
\begin{center}
    \begin{tabular}{l *{7}{c}}\hline
    $d=1$ & $2$ & $3$ & $4\ldots 10$  \\ \hline
  $\lambda=3.91$   & $4.78$ & $4.88$  & $4.88$
\end{tabular}
\end{center}
\label{table:simple_lyap1:cont1}
\end{table}

\begin{table}{}
\caption{Maximum $\lambda$ as a function of polynomial degree, $d_1=d_2=d$ for PDE~\eqref{eqn:exmp2} for which we can construct controllers using  with $K_1=K_2=P_1=P_2=0$.}
\vspace{-10pt}
\begin{center}
    \begin{tabular}{l *{7}{c}}\hline
    $d=1$ & $2$ & $3$ & $4\ldots 10$  \\ \hline
  $\lambda=3.51$   & $5.47$ & $6.64$  & $6.64$
\end{tabular}
\end{center}
\label{table:simple_lyap1:cont2}
\end{table}

Comparing Tables~\ref{table:simple_lyap1:cont1}-\ref{table:simple_lyap1:cont2} with Tables~\ref{table:exmp1:cont}-\ref{table:exmp2:cont} we observe that the inclusion of kernels $K_1$, $K_2$, $P_1$ and $P_2$ allows the construction of output feedback based controllers for significantly higher values of $\lambda$. Moreover, by setting $K_1=K_2=P_2=P_2=0$, the numerical results appear to show an upper bound to the $\lambda$ for which we can design controllers without the use of these kernel functions. We conjecture, therefore, that kernel functions are a necessary part of any Lyapunov-based method for analysis and control of PDEs.

Finally, as we previously stated, the SOS conditions for the design of output feedback controllers can be easily modified for systems with other types of boundary conditions. To this end, we provide the numerical results for controller synthesis for PDEs~\eqref{eqn:exmp1} and \eqref{eqn:exmp2} with boundary conditions defined in Table~\ref{table:alt_BC}.
\begin{table}{}
\caption{Alternative boundary conditions and outputs for PDE~\eqref{eqn:exmp1}-\eqref{eqn:exmp2}.}
\begin{center}
    \begin{tabular}{l *{3}{c}}\hline
  & Boundary Condition & Output $y(t)$  \\ \hline
   Dirichlet & {$\!\begin{aligned}
               w(0,t) &= 0 \\
               w(1,t) &= u(t) \end{aligned}$} & $w_x(1,t)$  \\ \hline
   Neumann & {$\!\begin{aligned}
               w_x(0,t) &= 0 \\
               w_x(1,t) &= u(t) \end{aligned}$} & $w(1,t)$ \\ \hline
   Robin & {$\!\begin{aligned}
               w(0,t)+w_x(0,t) &= 0 \\
               w(1,t)+w_x(1,t) &= u(t) \end{aligned}$} & $w(1,t)$
\end{tabular}
\end{center}
\label{table:alt_BC}
\end{table}
Table~\ref{table:alt:cont1} illustrates the maximum $\lambda$ for which we can construct output feedback controllers as a function of $d_1=d_2=d$ for PDE~\eqref{eqn:exmp1} with boundary conditions and outputs given in Table~\ref{table:alt_BC}. Similarly, Table~\ref{table:alt:cont2} illustrates these results for PDE~\eqref{eqn:exmp2}.

\begin{table}{}
\caption{Maximum $\lambda$ as a function of polynomial degree, $d_1=d_2=d$ for PDE~\eqref{eqn:exmp1} with boundary conditions and outputs given in Table~\ref{table:alt_BC}.}
\begin{center}
    \begin{tabular}{l *{7}{c}}\hline
  &$d=7$ & $8$ & $9$ & $10$ & $11$  \\ \hline
  Dirichlet & $\lambda=14.25$ & $17.96$ & $17.96$ & $24.21$ & $25.78$   \\
  Neumann & $12.69$ & $16.01$ & $17.96$ & $17.96$ & $21.97$ \\
  Robin & $11.71$ & $14.45$ & $16.40$ & $17.96$ & $18.84$ \\
\end{tabular}
\end{center}
\label{table:alt:cont1}
\end{table}

\begin{table}{}
\caption{Maximum $\lambda$ as a function of polynomial degree, $d_1=d_2=d$ for PDE~\eqref{eqn:exmp2} with boundary conditions and outputs given in Table~\ref{table:alt_BC}.}
\begin{center}
    \begin{tabular}{l *{7}{c}}\hline
  &$d=4$ & $5$ & $6$ & $7$ & $8$  \\ \hline
  Dirichlet & $\lambda=21.87$ & $33.59$ & $36.71$ & $36.71$ & $44.53$   \\
  Neumann & $18.75$ & $29.78$ & $32.03$ & $32.03$ & $39.16$ \\
  Robin & $14.16$ & $26.66$ & $28.90$ & $28.90$ & $30.46$ \\
\end{tabular}
\end{center}
\label{table:alt:cont2}
\end{table}

We note that the backstepping method has been applied to Example~\eqref{eqn:exmp1} and is also able to construct exponentially stabilizing output feedback boundary controllers for arbitrary $\lambda>0$ (see~\cite{krstic2008adaptive}). Therefore, while we cannot necessarily claim any improvement in performance over this established methods, our approach is at least competitive and may have certain advantages such as relative ease of implementation (changing the system is a one-line edit) and the fact that our approach does not require numerical integration of a PDE.
\section{CONCLUSIONS}
In this paper we developed an algorithmic approach for designing output feedback boundary controllers for a class of linear scalar valued inhomogeneous parabolic PDEs. Our approach is based on a parameterization of positive multiplier and integral operators with semi-separable kernels. We tested the approach on homogeneous and inhomogeneous PDEs using several different types of boundary feedback and several different types of point measurements. Furthermore, we tested our approach with and without kernel functions to determine if kernel functions are a necessary part of Lyapunov theory for PDEs. Our numerical results indicate that kernel functions are a necessary part of Lyapunov functions for PDEs. Further, our numerical results indicate there is little or no conservativity in the method and that our approach is competitive with well-established approaches such as backstepping.
Note that as yet, the observer-based controllers in this paper are not optimal in any norm. Therefore, an obvious future direction of this work is to extend our approach to $\mcl{H}_\infty$-optimal control.
\section*{APPENDIX}

To prove Lemmas~\ref{lem:dual_LOI} and~\ref{lem:primal_LOI}, we use the following identity.
\begin{lemma}[\cite{hardy1952inequalities},\cite{krstic2008boundary}]
\label{lem:wirtinger}
 let $w \in H^2(0,1)$ be a scalar function. Then
  \[\int_0^1 (w(x)-w(0))^2dx \leq \frac{4}{\pi^2} \int_0^1 w_x(x)^2 dx.\]
\end{lemma}

\begin{proof}[Lemma~\ref{lem:dual_LOI}]
We begin by considering the following decomposition
\begin{align}
&\ip{\mathcal{A}\pop z}{z}+\ip{z}{\mathcal{A}\pop z}  \nonumber \\
&= 2 \igzo \left(a(x)\frac{\partial^2}{\partial x^2}\left[(\pop z)(x) \right]+b(x)\pfx \left[(\pop z)(x) \right]\right)z(x)dx \nonumber \\
& \quad +\igzo  c(x)(\pop z)(x)  z(x) dx \nonumber \\
&\label{eqn:dual_gamma}= 2 \left(\Gamma_1+\Gamma_2+\Gamma_3+\Gamma_4+\Gamma_5  \right),
\end{align} where
\begin{align*}
\Gamma_1 =& \igzo z(x)a(x) \frac{\partial^2}{\partial x^2} \left[M(x)z(x) \right]dx,\\
\Gamma_2 =& \igzo z(x)b(x) \pfx \left[M(x)z(x) \right]dx,\\
\Gamma_3 =& \igzo z(x)a(x)\frac{\partial^2}{\partial x^2} \left(\igzx K_1(x,\xi)z(\xi)d\xi\right)dx \\
&+\igzo z(x)a(x)\frac{\partial^2}{\partial x^2} \left(\igxo K_2(x,\xi)z(\xi)d\xi\right)dx \\
\Gamma_4 =&\igzo z(x)b(x) \pfx \left(\igzx K_1(x,\xi)z(\xi)d\xi\right)dx\\
&+\igzo z(x)b(x) \pfx \left(\igxo K_2(x,\xi)z(\xi)d\xi\right)dx,\\
\Gamma_5 = &\igzo z(x)^2 M(x)c(x)dx\\
&+ \igzo \igzx z(x)c(x)K_1(x,\xi)z(\xi)d\xi dx \\
&+ \igzo \igxo z(x)c(x)K_2(x,\xi)z(\xi)d\xi dx,
\end{align*} where we have used the fact that
\[
K(x,\xi) = \begin{cases} K_1(x,\xi) & \xi \leq x \\
K_2(x,\xi) &  \xi>x \end{cases}.
\]

Before we proceed, we calculate the boundary condition at $x=0$. Since $z=\pinv w$, for any $w\in H^2(0,1)$ with $w_x(0)=0$, we have that $w=\pop z$. Using the definition of $\pop$,
\begin{align*}
w_x(0)=M_x(0)z(0)+M(0)z_x(0)+\igzo K_{2,x}(0,x)z(x)dx,
\end{align*} where we have used the fact that $K_1(x,\xi)=K_2(\xi,x)$. Since $w_x(0)=0$, we conclude that
\begin{equation}\label{eqn:dual:BC}
-M(0)z_x(0)=M_x(0)z(0)+\igzo K_{2,x}(0,x)z(x)dx.
\end{equation} Applying integration by parts twice and using the boundary condition at $x=0$, we get
\begin{align*}
\Gamma_1 =& -\igzo z_x(x)^2 a(x)M(x)dx \\
&+ \frac{1}{2}\igzo z(x)^2 \left[a_{xx}(x)M(x)+a(x)M_{xx}(x) \right]dx \\
&+ \frac{1}{2}z(1)^2 \left[a(1)M_x(1)-a_x(1)M(1) \right] \\
&+ \frac{1}{2} z(0)^2 \left[a(0)M_x(0)+a_x(0)M(0) \right] \\
&+z(1)a(1)M(1)z_x(1)+z(0)\igzo a(0)K_{2,x}(0,x)z(x)dx.
\end{align*}
 
From Theorem~\ref{thm:jointpos} it is readily established that $M(x) \geq \epsilon_1$. Additionally, we have that $a(x) \geq \alpha$. Therefore, $a(x)M(x) \geq \alpha \epsilon_1$ and we may apply Lemma~\ref{lem:wirtinger} to produce
\[-\igzo z_x(x)^2 a(x)M(x)dx \leq -\frac{\pi^2}{4}\alpha \epsilon_1 \igzo \left(z(x)-z(0)\right)^2 dx. \]
Therefore, \begin{align}
\Gamma_1 \leq & \igzo z(x)^2 \left( \frac{1}{2}\left[a_{xx}(x)M(x)+a(x)M_{xx}(x) \right]-\frac{\pi^2}{4}\alpha \epsilon_1 \right)dx  \nonumber \\
&+ \frac{1}{2}z(1)^2 \left[a(1)M_x(1)-a_x(1)M(1) \right] \nonumber \\
&+ \frac{1}{2} z(0)^2 \left[a(0)M_x(0)+a_x(0)M(0)-\frac{\pi^2}{2}\alpha \epsilon_1 \right] \nonumber \\
&+z(0)\igzo \left(a(0)K_{2,x}(0,x)+\frac{\pi^2}{2}\alpha \epsilon_1 \right)z(x)dx \nonumber \\
&\label{eqn:dual_gamma_1}+z(1)a(1)M(1)z_x(1).
\end{align}
Similarly, applying integration by parts once,
\begin{align}
\Gamma_2 = & \frac{1}{2}\igzo z(x)^2 \left[b(x)M_x(x)-b_x(x)M(x) \right]dx \nonumber \\
&\label{eqn:dual_gamma_2}+\frac{1}{2}z(1)^2 b(1)M(1)-\frac{1}{2}z(0)^2 b(0)M(0).
\end{align}
Now, note that for $(M,K_1,K_2) \in \Xi_{d_1,d_2,\epsilon_1,\epsilon_2}$, we have $K_1(x,\xi)=K_2(\xi,x)$ and thus $K_1(x,x)=K_2(x,x)$. Utilizing this property and applying integration by parts twice
\begin{align}
\Gamma_3 = & \igzo z(x)^2 \left(a(x)\left[\pfx[K_1(x,\xi)-K_2(x,\xi)] \right]_{\xi=x} \right)dx \nonumber \\
&+\igzo \igzx z(x)a(x)K_{1,xx}(x,\xi)z(\xi)d\xi dx \nonumber \\
&+\igzo \igxo z(x)a(x)K_{2,xx}(x,\xi)z(\xi)d\xi dx \nonumber
\end{align} Dividing the double integrals,
\begin{align}
\Gamma_3 = & \igzo z(x)^2 \left(a(x)\left[\pfx[K_1(x,\xi)-K_2(x,\xi)] \right]_{\xi=x} \right)dx \nonumber \\
&+\hlf \igzo \igzx z(x)a(x)K_{1,xx}(x,\xi)z(\xi)d\xi dx \nonumber \\
&+\hlf \igzo \igxo z(x)a(x)K_{2,xx}(x,\xi)z(\xi)d\xi dx \nonumber \\
&+\hlf \igzo \igzx z(x)a(x)K_{1,xx}(x,\xi)z(\xi)d\xi dx \nonumber \\
&+\hlf \igzo \igxo z(x)a(x)K_{2,xx}(x,\xi)z(\xi)d\xi dx \nonumber
\end{align} Changing the order of integration, switching between $x$ and $\xi$ and using the fact that $K_1(x,\xi)=K_2(\xi,x)$ in the last two double integral produces
 \begin{align}
&\Gamma_3 \nonumber \\
&=  \igzo z(x)^2 \left(a(x)\left[\pfx[K_1(x,\xi)-K_2(x,\xi)] \right]_{\xi=x} \right)dx \nonumber \\
& \quad + \hlf \igzo \igzx z(x)\left[a(x)K_{1,xx}(x,\xi)+a(\xi)K_{1,\xi \xi}(x,\xi)\right]z(\xi)d\xi dx \nonumber \\
&\label{eqn:dual_gamma_3} \quad + \hlf \igzo \igxo z(x)\left[a(x)K_{2,xx}(x,\xi)+a(\xi)K_{2,\xi \xi}(x,\xi)\right]z(\xi)d\xi dx.
\end{align} Applying integration by parts once and employing the previously performed change of order of integration
 \begin{align}
 &\Gamma_4 \nonumber \\
 &=\hlf \igzo \igzx z(x) \left[b(x)K_{1,x}(x,\xi)+b(\xi)K_{1,\xi}(x,\xi) \right] z(\xi)z(x)d\xi dx \nonumber \\
 &\label{eqn:dual_gamma_4} \quad +\hlf \igzo \igxo z(x) \left[b(x)K_{2,x}(x,\xi)+b(\xi)K_{2,\xi}(x,\xi) \right] z(\xi)z(x)d\xi dx.
 \end{align} Finally, applying a change of order of integration as applied to $\Gamma_3$ and $\Gamma_4$,
\begin{align}
\Gamma_5=& \igzo z(x)^2 M(x)c(x)dx \nonumber \\
& +\hlf \igzo \igzx z(x) \left[c(x)+c(\xi)\right]K_1(x,\xi)z(\xi)d\xi dx \nonumber \\
&\label{eqn:dual_gamma_5} +\hlf \igzo \igxo z(x) \left[c(x)+c(\xi)\right]K_2(x,\xi)z(\xi)d\xi dx
\end{align} Substituting \eqref{eqn:dual_gamma_1}-\eqref{eqn:dual_gamma_5} in \eqref{eqn:dual_gamma} and using Definition~\ref{def:dual},
\begin{align*}
&\ip{\mcl{A}\pop z}{z}+\ip{z}{\mcl{A}\pop z} \\
&\leq  \igzo z(x)^2 T_0(x)dx + \igzo \igzx z(x)T_1(x,\xi)z(\xi)d\xi dx \\
& \quad  + \igzo \igxo z(x)T_2(x,\xi)z(\xi)d\xi dx \\
&\quad +z(0)\left(T_3z(0)+\igzo T_4(x)z(x)dx \right)\\
& \quad \quad  + z(1) \left(T_5 z(1) +T_6z_x(1) \right).
\end{align*}
Finally, using the definition of operator $\mcl{T}$,
\begin{align*}
&\ip{\mcl{A}\pop z}{z}+\ip{z}{\mcl{A}\pop z} \\
&\leq  \ip{\mcl{T}z}{z}+z(0)\left(T_3z(0)+\igzo T_4(x)z(x)dx \right)\\
& \quad \quad  + z(1) \left(T_5 z(1) +T_6z_x(1) \right).
\end{align*}
\end{proof}

\begin{proof}[Lemma~\ref{lem:primal_LOI}]
Using the self-adjointedness of operator $\sop$ we begin with the following decomposition
\begin{align}
&\ip{\mathcal{A}w}{\mathcal{S}w}+\ip{\mathcal{S}\mathcal{A}w}{w} \nonumber \\
&=2\ip{\mathcal{A}w}{\mathcal{S}w} \nonumber \\
 &=2\igzo \left(a(x)w_{xx}(x)+b(x)w_x(x)+c(x)w(x) \right)(\sop w)(x)dx \nonumber \\
&\label{Gamma_eqn} = 2\left(\Gamma_1+\Gamma_2+\Gamma_3+\Gamma_4+\Gamma_5\right),
\end{align}
 where
 \begin{align*}
 \Gamma_1 =& \igzo w_{xx}(x)a(x)N(x)w(x)dx, \\
 \Gamma_2 =& \igzo w_x(x)b(x)N(x)w(x)dx, \\
 \Gamma_3 =& \igzo w_{xx}(x)a(x) \igzx P_1(x,\xi)w(\xi) d\xi dx \\
 &+ \igzo w_{xx}(x)a(x)\igxo P_2(x,\xi)w(\xi)d \xi  dx,\\
 \Gamma_4 = &\igzo w_x(x)b(x) \igzx P_1(x,\xi)w(\xi) d\xi dx \\
 &+ \igzo w_x(x)b(x)\igxo P_2(x,\xi)w(\xi)d\xi dx,\\
 \Gamma_5 = & \igzo w(x)^2 N(x) c(x) dx \\
 &+ \igzo \igzx w(x) c(x) P_1(x,\xi)w(\xi) d\xi dx \\
 & + \igzo \igxo w(x) c(x) P_2(x,\xi)w(\xi) d\xi dx.
\end{align*} Here we have used the fact that
\[
P(x,\xi) = \begin{cases} P_1(x,\xi) & \xi \leq x \\
P_2(x,\xi) &  \xi>x \end{cases}.
\]

Applying integration by parts twice and using the boundary condition $w_x(0)=0$ yields
 \begin{align*}
 \Gamma_1 =& - \igzo w_x(x)^2 a(x)N(x)dx \\
 &+ \hlf \igzo  \frac{\partial^2}{\partial x^2} \left[ a(x)N(x)\right] w(x)^2 dx \\
 &-  \frac{1}{2}\left(a_x(1)N(1)+ a(1)N_x(1) \right) w(1)^2 \\
 &+\hlf \left(a_x(0)N(0)+a(0)N_x(0) \right)w(0)^2 \\
 &+ w_x(1)a(1)N(1) w(1).
 \end{align*}
 Since $a(x) \geq \alpha >0$ and $\{N,P_1,P_2\} \in \Xi_{d_1,d_2,\epsilon_1,\epsilon_2}$, we have $a(x)N(x) \geq \alpha \epsilon_1$. Thus, by application of Lemma~\ref{lem:wirtinger} we get
 \[
 - \igzo w_x(x)^2 a(x)N(x) dx \leq -\frac{\pi^2}{4}\alpha \epsilon_1 \igzo \left(w(x)-w(0)\right)^2 dx.
 \]
Therefore, we conclude that
 \begin{align}
 \Gamma_1 \leq & \hlf \igzo w(x)^2\left( \frac{\partial^2}{\partial x^2} \left[ a(x)N(x)\right] -\frac{\pi^2}{2} \alpha \epsilon_1\right) dx \nonumber \\
 &+\frac{\pi^2}{2}\alpha \epsilon_1 w(0)\igzo w(x)dx \nonumber \\ 
 &  - \frac{1}{2}\left(a_x(1)N(1)+ a(1)N_x(1) \right) w(1)^2 \nonumber \\
&+\hlf \left(a_x(0)N(0)+a(0)N_x(0)-\frac{\pi^2}{2}\alpha \epsilon_1 \right)w(0)^2 \nonumber \\ 
 &\label{Gamma_1}+w_x(1) a(1)N(1)w(1). 
\end{align}
Similarly, applying integration by parts once
 \begin{align}
 \Gamma_2 =&  - \hlf \igzo w(x)^2   \pfx \left[ b(x)N(x)\right] dx \nonumber \\
 &\label{Gamma_2} +   \frac{1}{2}b(1)N(1) w(1)^2-\hlf b(0)N(0)w(0)^2.
 \end{align}
Now, note that for $(N,P_1,P_2) \in \Xi_{d_1,d_2,\epsilon_1,\epsilon_2}$, we have $P_1(x,\xi)=P_2(\xi,x)$ and thus $P_1(x,x)=P_2(x,x)$. Exploiting this property and using the boundary condition, we may apply integration by parts twice and use $w_x(0)=0$ to obtain 
\begin{align}
 \Gamma_3  = &\igzo w(x)^2\left( \left[ \pfx \left[a(x)(P_1(x,\xi)-P_2(x,\xi)) \right] \right]_{\xi=x} \right) dx \nonumber \\
 &  +  \igzo w(x)\igzx  \left( \frac{\partial^2}{\partial x^2}\left[a(x)P_1(x,\xi) \right] \right) w(\xi) d\xi dx \nonumber \\
 &  +  \igzo w(x)\igxo  \left( \frac{\partial^2}{\partial x^2}\left[a(x)P_2(x,\xi) \right] \right) w(\xi) d\xi  dx \nonumber \\
 &   - w(1) \igzo \left(a_x(1)P_1(1,x)+a(1)P_{1,x}(1,x) \right) w(x) dx \nonumber \\
 & + w(0) \igzo \left(a_x(0)P_2(0,x)+a(0)P_{2,x}(0,x) \right)w(x)dx \nonumber \\
 &   + w_x(1) \igzo a(1)P_1(1,x) w(x)dx. \nonumber
 \end{align} We can divide the two double integrals as
\begin{align}
 \Gamma_3 \nonumber  = & \igzo w(x)^2\left( \left[ \pfx \left[a(x)(P_1(x,\xi)-P_2(x,\xi)) \right] \right]_{\xi=x} \right) dx \nonumber \\
 & + \hlf \igzo w(x)\igzx  \left( \frac{\partial^2}{\partial x^2}\left[a(x)P_1(x,\xi) \right] \right) w(\xi) d\xi dx \nonumber \\
 &  + \hlf \igzo w(x)\igxo  \left( \frac{\partial^2}{\partial x^2}\left[a(x)P_2(x,\xi) \right] \right) w(\xi) d\xi  dx \nonumber \\
&  + \hlf \igzo w(x)\igzx  \left( \frac{\partial^2}{\partial x^2}\left[a(x)P_1(x,\xi) \right] \right) w(\xi) d\xi dx \nonumber \\
 &  + \hlf \igzo w(x)\igxo  \left( \frac{\partial^2}{\partial x^2}\left[a(x)P_2(x,\xi) \right] \right) w(\xi) d\xi  dx \nonumber \\
 &   - w(1) \igzo \left(a_x(1)P_1(1,x)+a(1)P_{1,x}(1,x) \right) w(x) dx \nonumber \\
& + w(0) \igzo \left(a_x(0)P_2(0,x)+a(0)P_{2,x}(0,x) \right)w(x)dx \nonumber \\ 
 &   + w_x(1) \igzo a(1)P_1(1,x) w(x)dx. \nonumber
 \end{align}
  Changing the order of integration in the last two double integrals, switching the variables $x$ and $\xi$ and using $P_1(x,\xi)=P_2(\xi,x)$,
 \begin{align}
 \Gamma_3 = & \igzo w(x)^2\left( \left[ \pfx \left[a(x)(P_1(x,\xi)-P_2(x,\xi)) \right] \right]_{\xi=x} \right) dx \nonumber \\
 &  +  \igzo \igzx w(x) \left(\hlf \frac{\partial^2}{\partial x^2}\left[a(x)P_1(x,\xi) \right]\right)w(\xi)d\xi dx \nonumber \\
 &  +\igzo \igzx w(x)\left(\hlf \frac{\partial^2}{\partial \xi^2}\left[a(\xi)P_1(x,\xi) \right] \right) w(\xi) d\xi dx \nonumber \\
 & + \igzo \igxo w(x) \left(\hlf \frac{\partial^2}{\partial x^2}\left[a(x)P_2(x,\xi) \right]\right)w(\xi)d\xi dx \nonumber \\
 & +\igzo \igxo w(x) \left(\hlf \frac{\partial^2}{\partial \xi^2}\left[a(\xi)P_2(x,\xi) \right] \right) w(\xi) d\xi dx \nonumber \\
 &  - w(1) \igzo \left(a_x(1)P_1(1,x)+a(1)P_{1,x}(1,x) \right) w(x) dx \nonumber \\
 & + w(0) \igzo \left(a_x(0)P_2(0,x)+a(0)P_{2,x}(0,x) \right)w(x)dx \nonumber \\
 &\label{Gamma_3}   + w_x(1) \igzo a(1)P_1(1,x) w(x)dx.
 \end{align}
  Applying integration by parts once and following the same procedure as for $\Gamma_3$,
 \begin{align}
 \Gamma_4 =&- \igzo \igzx w(x) \left(\hlf \pfx \left[b(x)P_1(x,\xi)\right]\right)w(\xi)d\xi dx \nonumber \\
 &- \igzo \igzx w(x)\left(\hlf \frac{\partial}{\partial \xi}\left[ b(\xi)P_1(x,\xi)\right] \right) w(\xi) d\xi dx  \nonumber \\
 & - \igzo \igxo w(x) \left(\hlf \pfx \left[b(x)P_2(x,\xi) \right]\right)w(\xi)d\xi dx \nonumber \\
 &- \igzo \igxo w(x)\left(\hlf \frac{\partial}{\partial \xi} \left[b(\xi)P_2(x,\xi) \right] \right) w(\xi) d\xi dx \nonumber \\
 &  + w(1) \igzo b(1)P_1(1,x)w(x)dx \nonumber \\
 &\label{Gamma_4} -w(0)\igzo b(0)P_2(0,x)w(x)dx.
 \end{align} 
Finally, employing a change of order of integration as done for $\Gamma_3$ and $\Gamma_4$ produces
 \begin{align}
 \Gamma_5 = & \igzo w(x)^2 N(x) c(x) dx \nonumber \\
 &+ \igzo \igzx w(x) \left(\hlf \left[c(x)+c(\xi) \right] P_1(x,\xi) \right) w(\xi) d\xi dx \nonumber \\
 &\label{Gamma_5} + \igzo \igxo w(x) \left( \hlf \left[c(x)+c(\xi)\right] P_2(x,\xi)\right)w(\xi) d\xi dx.
 \end{align}
 Substituting \eqref{Gamma_1}-\eqref{Gamma_5} into \eqref{Gamma_eqn} and using Definition~\ref{def:primal} gives us
 \begin{align*}
 &\ip{\mathcal{A}w}{\mathcal{S}w}+\ip{\mathcal{S}\mathcal{A}w}{w} \\
 &\leq  \igzo w(x)^2 Q_0(x)dx + \igzo \igzx w(x)Q_1(x,\xi)w(\xi)d\xi dx\\
 & \quad  + \igzo \igxo w(x)Q_2(x,\xi)w(\xi)d\xi dx \\
 & \quad + w(0) \left(Q_3 w(0) + \igzo Q_4(x)w(x)dx \right) \\
 & \quad +w(1) \left(Q_5 w(1)+\igzo Q_6(x)w(x)dx \right) \\
 &\quad+ w_x(1) \left(Q_7 w(1)+\igzx Q_8(x)w(x)dx \right).
 \end{align*} 
 Finally, using the definition of operator $\mcl{Q}$,
 \begin{align*}
 &\ip{\mathcal{A}w}{\mathcal{S}w}+\ip{\mathcal{S}\mathcal{A}w}{w} \\
 &\leq  \ip{w}{\mcl{Q}w}  + w(0) \left(Q_3 w(0) + \igzo Q_4(x)w(x)dx \right) \\
 & \quad +w(1) \left(Q_5 w(1)+\igzo Q_6(x)w(x)dx \right) \\
 &\quad+ w_x(1) \left(Q_7 w(1)+\igzx Q_8(x)w(x)dx \right).
 \end{align*}
\end{proof}
\section*{ACKNOWLEDGMENT}

This research was carried out with the financial support of NSF CAREER Grant CMMI-1151018.


\bibliographystyle{plain}
\bibliography{CDC_015}

\begin{thebibliography}{10}

\bibitem{balogh2004stability}
A.~Balogh and M.~Krstic.
\newblock Stability of partial difference equations governing control gains in
  infinite-dimensional backstepping.
\newblock {\em Systems \& control letters}, 51(2):151--164, 2004.

\bibitem{blum1998complexity}
L.~Blum.
\newblock {\em {Complexity and real computation}}.
\newblock Springer Verlag, 1998.

\bibitem{dullerud2000course}
G.E. Dullerud and F.G. Paganini.
\newblock {\em {A course in robust control theory: a convex approach}}.
\newblock Springer Verlag, 2000.

\bibitem{fridman2009lmi}
E.~Fridman and Y.~Orlov.
\newblock An {LMI} approach to ${H}_\infty$ boundary control of semilinear
  parabolic and hyperbolic systems.
\newblock {\em Automatica}, 45(9):2060--2066, 2009.

\bibitem{gohberg1984time}
I.~Gohberg and M.~A. Kaashoek.
\newblock Time varying linear systems with boundary conditions and integral
  operators. {I}. {T}he transfer operator and its properties.
\newblock {\em Integral equations and Operator theory}, 7(3):325--391, 1984.

\bibitem{gu2003stability}
K.~Gu, V.~Kharitonov, and J.~Chen.
\newblock {\em Stability of time-delay systems}.
\newblock Birkhauser, 2003.

\bibitem{hardy1952inequalities}
G.~H. Hardy, J.~E. Littlewood, and G.~Polya.
\newblock {\em Inequalities}.
\newblock Cambridge university press, 1952.

\bibitem{kreyszig1989introductory}
E.~Kreyszig.
\newblock {\em Introductory functional analysis with applications}, volume~21.
\newblock {W}iley, 1989.

\bibitem{krstic2008adaptive}
M.~Krstic and A.~Smyshlyaev.
\newblock Adaptive boundary control for unstable parabolic {PDE}s, {P}art {I}:
  Lyapunov design.
\newblock {\em IEEE Transactions on Automatic Control}, 53(7):1575--1591, 2008.

\bibitem{krstic2008boundary}
M.~Krstic and A.~Smyshlyaev.
\newblock {\em Boundary control of PDEs: A course on backstepping designs},
  volume~16.
\newblock Society for Industrial Mathematics, 2008.

\bibitem{lasiecka2000control}
I.~Lasiecka and R.~Triggiani.
\newblock {\em Control theory for partial differential equations: Volume 1,
  Abstract parabolic systems: Continuous and approximation theories}, volume~1.
\newblock Cambridge University Press, 2000.

\bibitem{morris2010approximation}
K.~Morris and C.~Navasca.
\newblock Approximation of low rank solutions for linear quadratic control of
  partial differential equations.
\newblock {\em Computational Optimization and Applications}, 46(1):93--111,
  2010.

\bibitem{morris1994design}
K.~A. Morris.
\newblock Design of finite-dimensional controllers for infinite-dimensional
  systems by approximation.
\newblock 1994.

\bibitem{murray2002mathematical}
J.~D. Murray.
\newblock {\em Mathematical biology}, volume~2.
\newblock Springer, 2002.

\bibitem{papachristodoulou2005constructing}
A.~Papachristodoulou, M.~M. Peet, and S.~Lall.
\newblock Constructing {L}yapunov-{K}rasovskii functionals for linear time
  delay systems.
\newblock In {\em Proceedings of the 2005 American Control Conference}, pages
  2845--2850.

\bibitem{parrilo2000structured}
P.A. Parrilo.
\newblock {\em Structured semidefinite programs and semialgebraic geometry
  methods in robustness and optimization}.
\newblock PhD thesis, California Institute of Technology, 2000.

\bibitem{peetlmi}
M.~M. Peet.
\newblock L{M}{I} parametrization of {L}yapunov functions for
  infinite-dimensional systems: {A} framework.
\newblock In {\em American Control Conference (ACC), 2014}, pages 359--366.
  IEEE, 2014.

\bibitem{peet2008using}
M.~M. Peet and A.~Papachristodoulou.
\newblock Using polynomial semi-separable kernels to construct
  infinite-dimensional {L}yapunov functions.
\newblock In {\em 47th IEEE Conference on Decision and Control, 2008}, pages
  847--852.

\bibitem{prajna2001introducing}
S.~Prajna, A.~Papachristodoulou, and P.~A. Parrilo.
\newblock Introducing {SOSTOOLS}: A general purpose sum of squares programming
  solver.
\newblock In {\em Proceedings of the 41st IEEE Conference on Decision and
  Control, 2002}, volume~1, pages 741--746.

\bibitem{staffans1997quadratic}
O.~Staffans.
\newblock Quadratic optimal control of stable well-posed linear systems.
\newblock {\em Transactions of the American Mathematical Society},
  349(9):3679--3715, 1997.

\bibitem{staffans1998quadratic}
O.~J. Staffans.
\newblock Quadratic optimal control of well-posed linear systems.
\newblock {\em SIAM Journal on Control and Optimization}, 37(1):131--164, 1998.

\bibitem{sturm1999using}
J.~F. Sturm.
\newblock Using {S}e{D}u{M}i 1.02, a {MATLAB} toolbox for optimization over
  symmetric cones.
\newblock {\em Optimization methods and software}, 11(1-4):625--653, 1999.

\bibitem{triggiani1980boundary}
R.~Triggiani.
\newblock {Boundary feedback stabilizability of parabolic equations}.
\newblock {\em Applied Mathematics and Optimization}, 6(1):201--220, 1980.

\bibitem{van1993h}
B.~Van~Keulen.
\newblock {\em ${H}_\infty$-control for distributed parameter systems: a state
  space approach}.
\newblock Birkhauser, 1993.

\bibitem{weiss1997optimal}
M.~Weiss and G.~Weiss.
\newblock Optimal control of stable weakly regular linear systems.
\newblock {\em Mathematics of control, signals and systems}, 10(4):287--330,
  1997.

\bibitem{witrant2007control}
E.~Witrant, E.~Joffrin, S.~Br{\'e}mond, G.~Giruzzi, D.~Mazon, O.~Barana, and
  P.~Moreau.
\newblock A control-oriented model of the current profile in tokamak plasma.
\newblock {\em Plasma Physics and Controlled Fusion}, 49(7):1075, 2007.

\end{thebibliography}


\begin{thebibliography}{10}

\bibitem{Valmo_2}
M.~Ahmadi, G.~Valmorbida, and A.~Papachristodoulou.
\newblock Input-{O}utput {A}nalysis of {D}istributed {P}arameter {S}ystems
  {U}sing {C}onvex {O}ptimization.
\newblock In {\em 53rd Conference on Decision and Control (CDC)}. IEEE, 2014.

\bibitem{balogh2004stability}
A.~Balogh and M.~Krstic.
\newblock Stability of partial difference equations governing control gains in
  infinite-dimensional backstepping.
\newblock {\em Systems \& control letters}, 51(2):151--164, 2004.

\bibitem{blum1998complexity}
L.~Blum.
\newblock {\em {Complexity and real computation}}.
\newblock Springer Verlag, 1998.

\bibitem{coron2008dissipative}
J.~M. Coron, G.~Bastin, and B.~d'Andr{\'e}a Novel.
\newblock Dissipative boundary conditions for one-dimensional nonlinear
  hyperbolic systems.
\newblock {\em SIAM Journal on Control and Optimization}, 47(3):1460--1498,
  2008.

\bibitem{coron2007strict}
J.~M. Coron, B.~d'Andrea Novel, and G.~Bastin.
\newblock A strict {L}yapunov function for boundary control of hyperbolic
  systems of conservation laws.
\newblock {\em IEEE Transactions on Automatic Control}, 52(1):2--11, 2007.

\bibitem{fridman2009lmi}
E.~Fridman and Y.~Orlov.
\newblock An {LMI} approach to ${H}_\infty$ boundary control of semilinear
  parabolic and hyperbolic systems.
\newblock {\em Automatica}, 45(9):2060--2066, 2009.

\bibitem{gahlawat2011designing}
A.~Gahlawat and M.M. Peet.
\newblock Designing observer-based controllers for {PDE} systems: A
  heat-conducting rod with point observation and boundary control.
\newblock In {\em 50th IEEE Conference on Decision and Control and European
  Control Conference (CDC-ECC)}, pages 6985--6990. IEEE, 2011.

\bibitem{gohberg1984time}
I.~Gohberg and M.~A. Kaashoek.
\newblock Time varying linear systems with boundary conditions and integral
  operators. {I}. {T}he transfer operator and its properties.
\newblock {\em Integral equations and Operator theory}, 7(3):325--391, 1984.

\bibitem{hardy1952inequalities}
G.~H. Hardy, J.~E. Littlewood, and G.~Polya.
\newblock {\em Inequalities}.
\newblock Cambridge university press, 1952.

\bibitem{kamyar2014polynomial}
R.~Kamyar and M.~Peet.
\newblock Polynomial {O}ptimization with {A}pplications to {S}tability
  {A}nalysis and {C}ontrol-{A}lternatives to {S}um of {S}quares.
\newblock {\em arXiv:1408.5119}, 2014.

\bibitem{kreyszig1989introductory}
E.~Kreyszig.
\newblock {\em Introductory functional analysis with applications}, volume~21.
\newblock {W}iley, 1989.

\bibitem{krstic2008adaptive}
M.~Krstic and A.~Smyshlyaev.
\newblock Adaptive boundary control for unstable parabolic {PDE}s, {P}art {I}:
  Lyapunov design.
\newblock {\em IEEE Transactions on Automatic Control}, 53(7):1575--1591, 2008.

\bibitem{krstic2008boundary}
M.~Krstic and A.~Smyshlyaev.
\newblock {\em Boundary control of {PDE}s: A course on backstepping designs},
  volume~16.
\newblock Society for Industrial Mathematics, 2008.

\bibitem{murray2002mathematical}
J.~D. Murray.
\newblock {\em Mathematical biology}, volume~2.
\newblock Springer, 2002.

\bibitem{papachristodoulou2006analysis}
A.~Papachristodoulou and M.~M. Peet.
\newblock On the analysis of systems described by classes of partial
  differential equations.
\newblock In {\em 45th IEEE Conference on Decision and Control, 2006}, pages
  747--752.

\bibitem{parrilo2000structured}
P.A. Parrilo.
\newblock {\em Structured semidefinite programs and semialgebraic geometry
  methods in robustness and optimization}.
\newblock PhD thesis, California Institute of Technology, 2000.

\bibitem{peet2006positive}
M.~Peet, A.~Papachristodoulou, and S.~Lall.
\newblock Positive forms and stability of linear time-delay systems.
\newblock In {\em 45th IEEE Conference on Decision and Control, 2006}, pages
  187--193. IEEE, 2006.

\bibitem{peetlmi}
M.~M. Peet.
\newblock L{M}{I} parametrization of {L}yapunov functions for
  infinite-dimensional systems: {A} framework.
\newblock In {\em American Control Conference (ACC), 2014}, pages 359--366.
  IEEE, 2014.

\bibitem{peet2008using}
M.~M. Peet and A.~Papachristodoulou.
\newblock Using polynomial semi-separable kernels to construct
  infinite-dimensional {L}yapunov functions.
\newblock In {\em 47th IEEE Conference on Decision and Control, 2008}, pages
  847--852.

\bibitem{prajna2001introducing}
S.~Prajna, A.~Papachristodoulou, and P.~A. Parrilo.
\newblock Introducing {SOSTOOLS}: A general purpose sum of squares programming
  solver.
\newblock In {\em Proceedings of the 41st IEEE Conference on Decision and
  Control, 2002}, volume~1, pages 741--746.

\bibitem{Valmo_1}
G.~Valmorbida, M.~Ahmadi, and A.~Papachristodoulou.
\newblock Semi-{D}efinite {P}rogramming and {F}unctional {I}nequalities for
  {D}istributed {P}arameter {S}ystems.
\newblock In {\em 53rd Conference on Decision and Control (CDC)}. IEEE, 2014.

\bibitem{witrant2007control}
E.~Witrant, E.~Joffrin, S.~Br{\'e}mond, G.~Giruzzi, D.~Mazon, O.~Barana, and
  P.~Moreau.
\newblock A control-oriented model of the current profile in tokamak plasma.
\newblock {\em Plasma Physics and Controlled Fusion}, 49(7):1075, 2007.

\end{thebibliography}

\end{document}